\documentclass[letterpaper, DIV=14]{scrartcl}

\usepackage{amssymb, amsthm, mathtools,bbm}
\usepackage[square,sort,comma,numbers]{natbib}
\usepackage[none]{hyphenat}
\newtheorem{theorem}{Theorem}
\newtheorem{corollary}[theorem]{Corollary}
\newtheorem{lemma}[theorem]{Lemma}
\newtheorem{definition}[theorem]{Definition}
\newtheorem{proposition}[theorem]{Proposition}
\newtheorem{assumption}[theorem]{Assumption}
\numberwithin{equation}{section}
\numberwithin{theorem}{section}

\newcommand{\rr}{{\mathbb{R}}}

\newcommand{\nn}{{\mathbb{N}}}

\newcommand{\ee}{{\mathbb{E}\,}}

\newcommand{\pp}{{\mathbb{P}}}

\newcommand{\tr}{{\operatorname{Tr}\,}}

\newcommand{\beq}[1]{\begin{equation} \label{#1}}
\newcommand{\eeq}{\end{equation}}

\newcommand{\arcosh}{\operatorname{arcosh}}

\begin{document}
	\addtokomafont{author}{\raggedright}
	\title{ \raggedright The de Almeida-Thouless Line in Hierarchical Quantum Spin Glasses}
     \author{\hspace{-.075in} Chokri Manai and Simone Warzel}
     \date{\vspace{-.3in}}
	
	\maketitle

	\minisec{Abstract}
	We determine explicitly and discuss in detail the effects of the joint presence of a longitudinal and a transversal (random) magnetic field on the phases of the Random Energy Model (REM) and its hierarchical generalization, the GREM. 
	Our results extent known results both in the classical case of vanishing transversal field and in the quantum case for vanishing longitudinal field. 
	Following Derrida and Gardner, we argue that the longitudinal field has to be implemented hierarchically also in the Quantum GREM. We show that this ensures the shrinking of the spin glass phase in the presence of the magnetic fields as also expected for the Quantum Sherrington-Kirkpatrick model. 


\section{Introduction}\label{sec:intro}
Mean-field spin glasses such as the Sherrington-Kirkpatrick (SK) model  have long served as an inspiration to both physicists and mathematicians \cite{MPV86,Tal11,Pan13}. For these classical glasses, Parisi's replica ansatz for the free energy presents one of the rare gems of an exactly solvable case, whose solution covers extremely complex behaviour -- notably the occurrence of a frozen glass phase below a certain critical temperature $T_c $. 
Since spins are intrinsically quantum-mechanical objects, physicists have  started early on to investigate the quantum effects caused by the inclusion of a transversal magnetic field. 
Unfortunately, unlike the inclusion of a longitudinal magnetic field in the SK-model, the transversal field seems to crash all attempts of 
an explicit Parisi solution. One either has to resort to approximations or numerical calculations for the full phase diagram \cite{US87,YI86,HM93,ONS07,You17,SIC12} or bounds \cite{LRRS19,LMRW21} or more qualitative results~\cite{Craw07,AB19} for the Quantum SK-model. 
It is therefore rather remarkable that the associated hierarchical caricature, the generalised random energy model (GREM), still admits an explicit solution of Parisi type even in the presence of a transversal field \cite{Gold90,MW20a,MW20c}. 
The GREM was initially invented by Derrida~\cite{Der85} to qualitatively capture the behaviour of the free energy of more complicated glasses.  It was mathematically reformulated in~\cite{Rue87} and its significance for Parisi's ansatz was later clarified in~\cite{Gue03,ASS03,Tal06}. 

One central questions for spin glasses in external magnetic fields is whether the fields destabilise the low-temperature glass phase or not. For the SK-model in a constant longitudinal field, de Almeida and Thouless \cite{dAT78A} determined an equation for the critical temperature $T_c(h) $, which turns out to be decreasing in the field strength $ h $ and is known under the name de Almeida-Thouless (AT) line.  Below $T_c(h) $ the replica symmetry has been proven to be broken \cite{Ton02}.  Rigorous results above $  T_c(h) $ are still incomplete (see e.g.~\cite{AB+21} and refs.\ therein). 
Unlike for the SK-model, implementing the longitudinal field naively in GREM models causes the frozen phase to expand \cite{BK08,AK14,AP19}.
Derrida and Gardner \cite{DG86a} therefore suggested a  hierarchical implementation of the longitudinal magnetic field, which then leads again to a destabilisation of the frozen phase. 

The present paper now investigates the question of the stability of the low-temperature phase in general GREM models under the joint presence of a longitudinal and transversal field. We will present explicit formulas for the free energy of such Q(uantum)GREMs for both cases: a naive implementation of the longitudinal magnetic field and a hierarchical implementation. We will discuss the stability of the glass phase and calculate associated critical exponents. 

\subsection{The Quantum GREM with a random longitudinal field}\label{sec:qcremh}
The QGREM  with a (random) external transversal and longitudinal magnetic field is a Hamiltonian on $\psi \in \ell^2(\mathcal{Q}_N)$ 
of the form 
\begin{equation}\label{def:qrem}
(H_N\psi)({\pmb{\sigma}})= U({\pmb{\sigma}})\psi(({\pmb{\sigma}}) - h(\pmb{\sigma})\psi({\pmb{\sigma}})- (B \psi)(\pmb{\sigma}) .
\end{equation}
The first term represents the GREM energy landscape on the Hamming cube $\mathcal{Q}_N \coloneqq \{-1,1\}^N$ and 
is given by a centred Gaussian process $U(\pmb{\sigma})$
with covariance function 
	\begin{equation}\label{eq:uGREM}
\ee[U(\pmb{\sigma}) U(\pmb{\sigma}^\prime)] =  N A(q_N(\pmb{\sigma},\pmb{\sigma}^\prime)),
\end{equation}
where $A \coloneqq [0,1] \to [0,1]$ is a fixed non-decreasing, right-continuous, and normalised function, $A(1) =1$,
which does not depend on $N$. Moreover, $q_N$ denotes the  normalised lexicographic overlap of spin configurations 
$ \pmb{\sigma}, \pmb{\sigma}' \in \mathcal{Q}_N $:
\begin{equation}\label{eq:overlap}
q_N(\pmb{\sigma},\pmb{\sigma}^\prime) \coloneqq \begin{cases}
1 & \text{ if } \pmb{\sigma} = \pmb{\sigma}^\prime, \\
\frac1N \min \{1 \leq i \leq N; \sigma_i \neq \sigma^\prime_i \} & \text{ else }. \end{cases} 
\end{equation}
A straightforward implementation of a (random) longitudinal magnetic field is achieved through setting
\begin{equation}\label{eq:longfield}
h(\pmb{\sigma}) =  \sum_{j=1}^{N} h_j \sigma_j  .
\end{equation}
Interpreting the configuration basis $ \pmb{\sigma} $ as the $ z $-components of $ N $ quantum spin-$1/2 $,  
 a (random) transversal field $B$ in $ x $-direction is given by the sum of the Pauli $x$-matrices $  \pmb{s}_j $ with weights $b_j\in \mathbb{R} $:
\begin{equation}\label{def:B} 
(B\psi)(\pmb{\sigma}) :=    \sum_{j=1}^{N}  b_j \;  \big( \pmb{s}_j\psi\big)(\pmb{\sigma}) , \qquad  \big( \pmb{s}_j  \psi\big)(\pmb{\sigma}) :=   \psi(F_j \pmb{\sigma}), \qquad  F_j \pmb{\sigma} := (\sigma_1,\ldots,-\sigma_j,\ldots,\sigma_N) .
\end{equation}
We will assume throughout that the variables  $(U({\pmb{\sigma}}))  $, $ (h_j)$ and $(b_j)$ are mutually independent and that the field variables $h_j$ and $b_j$ are independent copies of absolutely integrable random variables $\mathfrak{h}$ and $\mathfrak{b}$, respectively.

Occurring phase transitions, in particular the de Almeida-Thouless line, are encoded in the limit of the pressure (or the negative free energy times the inverse temperature $ \beta $)
\begin{equation} \Phi_N(\beta,\mathfrak{h},\mathfrak{b}) \coloneqq \frac1N \ln \tr e^{-\beta H_N } \end{equation}
as the number of spins $ N $ goes to infinity.
Our first main theorem is an explicit formula for this limit  in terms of the concave hull $\bar{A}$ of $A$ and the right derivative $\bar{a}$ of $\bar{A}$.

\begin{theorem}\label{thm:qcremh}
Let $U(\pmb{\sigma})$ be a GREM with distribution function $A$ and suppose that the longitudinal random field is implemented as in~\eqref{eq:longfield}.  For any $\beta \geq 0$ and any absolutely integrable random variables $	\mathfrak{h},\mathfrak{b} $,  
the pressure converges  almost sure
\begin{equation}\label{eq:qcremh}
\lim_{N \to \infty} \Phi_N(\beta,\mathfrak{h},\mathfrak{b}) = \sup_{0 \leq z \leq 1} \left( \int_{0}^{z} \varphi(\beta,\mathfrak{h},x) \, dx + (1-z) \ee[\ln 2 \cosh(\beta \sqrt{\mathfrak{b}^2+\mathfrak{h}^2})] \right).
\end{equation}
The density $\varphi(\beta,\mathfrak{h},x)$ is given by 
\begin{equation}\label{eq:density}
\varphi(\beta,\mathfrak{h},x) \coloneqq \begin{cases} 
\ln 2 + \bar{a}(x)\frac{\beta^2}{2} + \ee[\ln \cosh \beta \mathfrak{h}] & if \quad \beta \leq \beta_c(x) \\
\beta(\bar{a}(x)\beta_c(x) + \ee[\mathfrak{h} \tanh \beta_c(x) \mathfrak{h}]) & if \quad \beta > \beta_c(x) 
\end{cases}
\end{equation}
where $\beta_c(x) = \beta_c(x,\mathfrak{h})$ is the unique positive solution of the self-consistency equation 
\begin{equation}\label{eq:bcrith}
\frac{\bar{a}(x)}{2}\beta_c(x)^2 =  \ln 2 + \ee[\ln \cosh \beta_c(x) \mathfrak{h}] - \beta_c(x) \ee[\mathfrak{h} \tanh \beta_c(x) \mathfrak{h}]. 
\end{equation}
Moreover, $\varphi(\beta,\mathfrak{h},x)$ is a decreasing function of $x$ and strictly increasing and convex in $\beta$, while $\beta_c(x)$ is increasing in $x$.
\end{theorem}

Theorem~\ref{thm:qcremh}, whose proof will be spelled out in Section~\ref{sec:pfqcremh} below, is a generalisation of Theorem~1.4 in~\cite{MW20c}, which addresses the case without a longitudinal field, $ \mathfrak{h} = 0 $. 
In the classical case without transversal magnetic  field, $ \mathfrak{b}= 0 $, it generalises the results of \cite{BK08}, which covers the case that $ \mathfrak{h}  $ is constant, and of  \cite{AK14,AP19}, which treats the special case of a REM or two-level GREM in a random magnetic field.

\subsection{Stability of the glass phase in the QGREM with longitudinal field}

From~$\eqref{eq:density}$ and the monotonicity of $\varphi(\beta,\mathfrak{h},x) $ and $\beta_c(x)$, it is evident that the location of the glass transition predicted by~\eqref{eq:qcremh} is completely determined by $\varphi(\beta,\mathfrak{h},0) $ which agrees with a rescaled REM pressure \cite{AK14}. 
The REM's energies $U(\pmb{\sigma})$, $ \pmb{\sigma} \in \mathcal{Q}_N $, are independent and identically distributed centred Gaussian variables with variance $ N $. This corresponds to choosing the step-function $ A(x) = 0 $ for $ x < 1 $ and $ A(1) = 1 $ in~\eqref{eq:uGREM}.  
In order to understand the qualitative behaviour of the phase diagram  and in particular the question of the stability of the glass phase in the QGREM with longitudinal field~\eqref{eq:longfield}, it is thus convenient to restrict the discussion to the REM with constant fields, i.e., $\mathfrak{h} = h$ and $\mathfrak{b} = \Gamma$ for some positive constants $h,\Gamma \geq 0$.  In fact, even quantitative properties such as the dependence of the critical temperature $T_c(h) =  \beta_c(0,h)^{-1} $ on the longitudinal field $ h $ coincide for the general GREM with the REM except for some numerical factors which depend on $\bar{a}(0)$. 
We therefore state the application of Theorem~\ref{thm:qcremh} to the QREM as our next corollary. 
\begin{corollary}\label{cor:qremh}
Consider a REM process $U(\pmb{\sigma})$ and constant longitudinal and transversal fields of strength $h ,\Gamma \geq 0$. Then, almost surely
\begin{equation}\label{eq:qremh}
\lim_{N \to \infty} \Phi_N(\beta,h,\Gamma) = \max\{ \Phi^\mathrm{REM}(\beta,h), \ln 2 \cosh(\beta \sqrt{h^2+\Gamma^2)} . \} ,
\end{equation}
where, $ \Phi^\mathrm{REM}(\beta,h)$ denotes the function 
\begin{equation}\label{eq:remh}
 \Phi^\mathrm{REM}(\beta,h) = \begin{cases} 
 \ln 2 + \frac{\beta^2}{2} + \ln \cosh \beta h ] & if \quad \beta \leq \beta_c(h) \\
 \beta(\beta_c(x) + h \tanh( \beta_c(h)h) & if \quad \beta > \beta_c(h) 
 \end{cases}
\end{equation}
and $\beta_c(h)$ is the unique positive solution of 	
\begin{equation}\label{eq:bcrit2}
\beta_c(h)^2 = 2 r(\tanh(\beta_c(h)h))
\end{equation}
with the modified binary entropy $r \colon [-1,1] \to \rr$,
\begin{equation}\label{eq:ent}
r(x) \coloneqq - \left(\frac{1-x}{2} \ln \frac{1-x}{2} + \frac{1+x}{2} \ln \frac{1+x}{2} \right).
\end{equation}
\end{corollary}

The short proof of Corollary~\ref{cor:qremh} can be found in the appendix.

\begin{figure}[ht]
	\begin{center}
		\includegraphics[width=.42\textwidth] {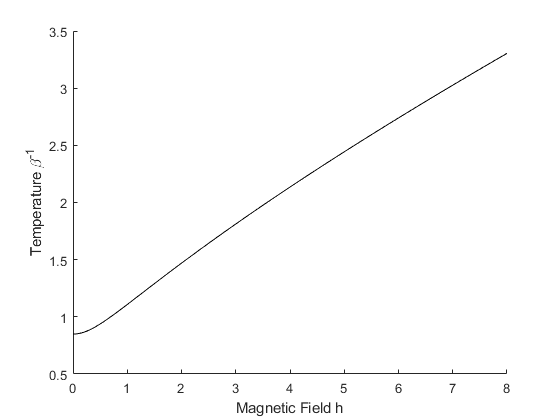}\hfill
		\includegraphics[width=.52\textwidth] {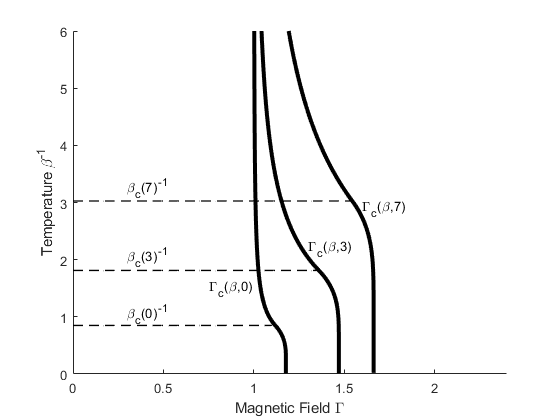}
		\caption{The left figures illustrates the freezing temperature $T_c(h) = \beta_c^{-1}(h)$ as a function of the longitudinal field $h$. On the right  is the $T-\Gamma$ phase diagram with the critical magnetic field $\Gamma_c(\beta,\Gamma)$ as well as the critical temperature evaluated at $h = 0,3,7$ }\label{fig:qrem}

	\end{center}
\end{figure}

For fixed $ h > 0 $ the phase diagram, which is plotted in Figure~\ref{fig:qrem},  resembles that of the QREM without longitudinal field~\cite{Gold90,MW20a}. 
The model undergoes a magnetic transition at 
\begin{equation}\label{eq:gamcrith}
\Gamma_c(\beta,h) \coloneqq \sqrt{\beta^{-2} \arcosh\left(\frac12 \exp(\Phi^{\mathrm{REM}}(\beta,h))\right)^2 -h^2 },
\end{equation}
where the magnetization in $x$-direction jumps. At fixed $ h > 0 $, this line separates the quantum paramagnet characterised by a positive magnetisation in $ x $-direction, from the classical spin glass. 

The unique positive  solution $ \beta_c(h) \in (0,\sqrt{2\ln 2 } ) $ of the self-consistency equation~\eqref{eq:bcrit2} marks the inverse freezing temperature at longitudinal field $ h > 0 $. 
For fixed $ h > 0 $ and $\Gamma < \Gamma_c(\beta,h)$  this line separates the high-temperature regime of the classical paramagnet from the spin glass phase. 
In comparison to the case $ h = 0 $, the longitudinal field causes an extensive magnetization $ M(\pmb{\sigma}) := \sum_{i=1}^N \sigma_i  $  in $z$-direction under the Gibbs average. 
The specific magnetization in $ z $-direction is a self-averaging quantity which converges  as $ N \to \infty $ to 
\[m_z(\beta,h) \coloneqq \frac{1}{\beta} \frac{\partial  \Phi}{\partial h}(\beta,h) =  \begin{cases} \tanh(\min\{\beta,\beta_c(h) \} h) , & \Gamma   < \Gamma_c(\beta,h) , \\ 
												\frac{h}{\sqrt{h^2+\Gamma^2}} \tanh( \beta \sqrt{h^2+\Gamma^2}) , & \Gamma   > \Gamma_c(\beta,h) .
												\end{cases}
\] 
The kink in its dependence on $ \beta $ for $  \Gamma   < \Gamma_c(\beta,h)  $ reflects the second-order freezing transition at $ \beta_c(h) $.

The following proposition summarises some basic properties of the critical inverse temperature $ \beta_c(h)$ and the critical transversal field $ \Gamma_c(\beta,h)$ as functions of $ h $.

\begin{proposition}\label{prop:critbg}
	The critical inverse temperature $\beta_c(h)$ and the critical magnetic field strength $\Gamma_c(\beta,h)$ have the following properties:
	\begin{itemize}
		\item[1.] $\beta_c(h)$ is a strictly decreasing function. Moreover, $ \beta_c(h) = \sqrt{2 \ln 2} \  (1 - h^2/2) +\mathcal{O}(h^4) $ for small $h$
			and asymptotically $\lim_{h \to \infty} \frac{h \beta_c(h)}{\ln h} = 1 $.
		\item[2.]  The high temperature limit
		$ \Gamma_c(0,h) \coloneqq \lim_{\beta  \to 0} \Gamma_c(\beta,h) = 1 $
		does not depend on $h$, and the low temperature limit
		\[ \lim_{\beta  \to \infty} \Gamma_c(\beta,h) = 
		\sqrt{ (\beta_c(h) + \tanh(\beta_c(h)h)h)^2 - h^2} \]
		resembles the ground-state phase transition.
		\item[3.] For any $\beta > 0$ the critical field strength $\Gamma_c(\beta,\cdot)$ is a strictly increasing function. In addition, we asymptotically have 
		$ \lim_{h \to \infty} \frac{\Gamma(\beta,h)}{\sqrt{h \beta_c(h)}} = 1 $.
	\end{itemize}
\end{proposition}

The proof of Proposition \ref{prop:critbg} is based on multiple elementary, but quite lengthy, computations, which we spelled out in the appendix for the convenience of the reader.\\

Let us put these findings in a general context. 
In classical SK-type models, the freezing temperature $T_c(h) = \beta_c(h)^{-1} $ decreases as $h$ becomes larger, i.e.~the glass phase shrinks \cite{dAT78A,Ton02}. Numerical calculations support the conjecture that in the Quantum SK-model, the longitudinal and transversal field destabilise the glass phase as well (cf.~\cite{ONS07,You17} and  \cite{SIC12}). In contrast, the REM and the QREM exhibit an expanding frozen phase for $h > 0$. This concerns not only the critical temperature $ T_c(h) $ but also the critical transversal magnetic field strength $\Gamma_c(\beta,h)$, which also increases with $ h $; see Figure~\ref{fig:qrem}. In this sense the QREM, although the limit $ p \to \infty $ of $ p $-spin models (cf.\ \cite{MW20b}),  features unphysical characteristics in presence of a longitudinal field. As we will argue next, this is a consequence of the unrealistic lack of correlations.

\subsection{The QGREM with a hierarchical longitudinal field }\label{sec:QGREMh}

That a longitudinal field stabilises the frozen phase in the QREM and QGREM, can be regarded as a quite unphysical behaviour. We will bypass this problem by following
 Derrida and Gardner's approach to incorporate the  magnetic field in $z$-direction as a hierarchical operator \cite{DG86a}. This choice can be physically justified: one should recall that the GREM was designed as a hierarchical approximation of the more involved SK-model, whose energy correlations are given by $ \mathbb{E}\left[ U(\pmb{\sigma}) U(\pmb{\sigma}^\prime) \right] = N r_N(\pmb{\sigma},\pmb{\sigma}^\prime)^2 $ in terms of the spin overlap $r_N(\pmb{\sigma},\pmb{\sigma}^\prime) = \sum_{j=1}^{N} \sigma_j \sigma_j^\prime$. 
 In fact, requiring that the entropy of likewise pair-correlated energies asymptotically coincides in the SK-model and the GREM, i.e.
 \[ \lim_{N \to \infty} \frac1N \ln\left( \frac{ \left| \{\pmb{\sigma}: r_N(\pmb{\sigma},\pmb{\sigma}^0)^2 > a  \} \right|   }{ \left| \{\pmb{\sigma}: A(q_N(\pmb{\sigma},\pmb{\sigma}^0))  > a  \} \right| }  \right) =1 \]
for all $ a \in (0,1) $, determines the choice $ A(x) = \gamma(x)^2 $,  where $\gamma$ is the inverse function  of 
\begin{equation}\label{eq:gamma}
\gamma^{-1}: [0,1] \to [0,1] , \quad \gamma^{-1}(a) \coloneqq 1- \frac{r(a)}{ \ln2}  = \frac{1-x}{2 \ln 2 } \ln (1-x) + \frac{1+x}{2\ln2} \ln (1+x) 
\end{equation} with the binary entropy $r$ from \eqref{eq:ent}. 
This follows from the known asymptotics 
\[  \left|\{\pmb{\sigma}: r_N(\pmb{\sigma},\pmb{\sigma}^0) > a/h  \}\right| \simeq 2^N 2^{-N\gamma^{-1}(a/h)}    \]
and $ \left| \{\pmb{\sigma}: q_N(\pmb{\sigma},\pmb{\sigma}^0) > a  \} \right| \simeq 2^N 2^{-aN} $. 

If we want to understand the SK-model with a longitudinal field, it is reasonable to consider the hierarchical reorganization of the magnetic field as well.
We start by introducing the notion of a general hierarchical field on the Hamming cube $\mathcal{Q}_N$.

\begin{definition}\label{def:hreg}
	We call a function $h \colon \mathcal{Q}_N \to \rr$ a hierarchical field with reference state $\pmb{\sigma}^0 \in \mathcal{Q}_N$ if there exists a function $\eta \colon [0,1] \to \rr$ such that 
	\begin{equation}
	h(\pmb{\sigma}) = N \eta(q_N(\pmb{\sigma},\pmb{\sigma}^0)),
	\end{equation}
	where $q$ is the lexicographic overlap~\eqref{eq:overlap}. Furthermore, $h$ is said to be a regular hierarchical field, if $\eta$ is a regular function on $[0,1]$, i.e.~$\eta$ is a uniform limit of step functions.
\end{definition}

Our second main result in this paper deals with general regular hierarchical fields. Nevertheless, let us in particular discuss the choice of $\pmb{\sigma}^0$ and $v$ that corresponds to a constant external magnetic field. To do so, we rewrite the original constant longitudinal magnetic field as follows 
\begin{equation}\label{eq:relh} h \sum_{i=1}^{N}\sigma_i  = hN r_N(\pmb{\sigma},\pmb{\sigma}^0)  \end{equation}
where  $\pmb{\sigma}^0 = (+1,\ldots,+1)$ is the ferromagnetic state. In the hierarchical case one may also think of $\pmb{\sigma}^0$ being the ferromagnetic state, but the free energy in fact does not depend on this reference state. 

Determining the "correct" overlap function is a little more subtle. One might be tempted to pick $\eta(q) = hq$ which yields the analogous relation between the field and the respective overlap as in~\eqref{eq:relh}. Similarly as discussed above, it is more reasonable though to demand that the entropy agrees, i.e.\ the number of (positive) energy states agree on an exponential scales
\[ \lim_{N \to \infty} \frac1N \ln\left( \frac{ \left| \{\pmb{\sigma}: h  r_N(\pmb{\sigma},\pmb{\sigma}^0) > a  \} \right|   }{ \left| \{\pmb{\sigma}: v(q_N(\pmb{\sigma},\pmb{\sigma}^0))  > a  \} \right| }  \right) =1 \]
for any $0 < a < h$. Comparing asymptotics  leads to the choice 
\begin{equation}\label{eq:maghierar} 
  \eta(a) \coloneqq   h \gamma(a ),
 \end{equation}
 where again $\gamma$ is the inverse function  of \eqref{eq:gamma}. Let us  record this as a definition:
 \begin{definition}
 	We call $h(\pmb{\sigma}) = N \eta(q_N(\pmb{\sigma},\pmb{\sigma}^0))$ with reference state $\pmb{\sigma}^0 = (+1,\ldots,+1)$ and overlap function $\eta$ given by \eqref{eq:maghierar} the hierarchical magnetic field of strength $h$. 
 \end{definition}

Our aim in the following is to determine the limit of the pressure 
$ \Phi_N(\beta,\mathfrak{b},h) $
of a Quantum GREM~\eqref{def:qrem} 
where $U$ is a GREM-type random process characterized by $ A $ in \eqref{eq:uGREM}, $h$ is a regular hierarchical field in the sense of Definition~\ref{def:hreg}, and $B$ is a random transversal field whose weights $b_j$ are independent copies of an absolutely integrable variable $\mathfrak{b}$ (see \eqref{def:B}).

To formulate our main result, we need to introduce doubly-cut GREM processes $U^{(y,z)}$ for $0 \leq y \leq z \leq 1$ on the reduced Hamming cube $\mathcal{Q}_{\lceil (z-y)N \rceil}$ with the 
(not normalized) distribution function $A^{(q,z)} \colon [0,z-y] \to [0,1]$,
$ 
A^{(y,z)}(x) \coloneqq A(x+y) - A(y) $. 
The corresponding concave hull and its right derivative are denoted by $\bar{A}^{(y,z)}$ and $\bar{a}^{(y,z)}$. We further set $\varphi^{(y,z)} \colon \rr \times [0,z-y] \to \rr$,
\begin{equation}\label{eq:varphi}
\varphi^{(y,z)}(\beta,x) \coloneqq \beta \sqrt{(2 \ln2)\ \bar{a}^{(y,z)}(x)}  \mathbbm{1}_{x < x^{(y,z)}(\beta)} + 
\left(\frac{\beta^2}{2} \bar{a}^{(y,z)}(x) + \ln2\right)  \mathbbm{1}_{x \geq x^{(y,z)}(\beta)}.
\end{equation}
with 
\begin{equation}
x^{(y,z)}(\beta) \coloneqq \sup \left\{x \ | \ \bar{a}^{(y,z)}(x) > 2 \ln2/ \beta^2 \right\}
\end{equation}
With these preparations we recall from Theorem 1.4 and Theorem 2.8 in \cite{MW20c}  that almost surely
\begin{align}\label{eq:QGREMh=0}
\lim_{N \to \infty} \Phi_N(\beta,\mathfrak{b},0) &=   \sup_{0 \leq z \leq 1} \left[ \int_{0}^{z} \varphi^{(0,1)}(\beta,x) \, dx + (1-z) \ee[\ln 2 \cosh(\beta \mathfrak{b})] \right] \notag \\
& = \sup_{0 \leq z \leq 1} \left[ \int_{0}^{z} \varphi^{(0,z)}(\beta,x) \, dx + (1-z) \ee[\ln 2 \cosh(\beta \mathfrak{b})] \right] .
\end{align}
In the presence of any regular hierarchical field $ h $ (not necessarily with $ \eta $ given by \eqref{eq:maghierar}), this result generalizes as follows. 
\begin{theorem}\label{thm:hierar}
Let $U(\pmb{\sigma})$ be of GREM and $B$ a random transversal field with independent weights~$(b_j)$ sharing the same distribution as $\mathfrak{b}$. Further, let $h(\pmb{\sigma}) = N \eta(q(\pmb{\sigma},\pmb{\sigma}^0))$ be a regular hierarchical field. Then, almost surely:
\begin{equation}\label{eq:hierar}
\begin{split}
\Phi(\beta,\mathfrak{b},h) & \coloneqq \lim_{N \to \infty} \Phi_N(\beta,\mathfrak{b},h) \\ &=\sup_{0 \leq y \leq z \leq 1}  \left[\beta \eta(y) + \int_{0}^{z-y} \varphi^{(y,z)}(\beta,x) \, dx + (1-z) \ee[\ln 2 \cosh(\beta \mathfrak{b})] \right]\\
&= \sup_{0 \leq y \leq z \leq 1} \left[ \beta \eta(y) + \int_{0}^{z-y} \varphi^{(y,1)}(\beta, x) \, dx + (1-z) \ee[\ln 2 \cosh(\beta \mathfrak{b})] \right] .
\end{split}
\end{equation}
\end{theorem}

Remarkably, the transversal field $B$ and the hierarchical field $h$ affect the glass phase quite differently. While the hierarchical field tends to shrink the glass region in its most correlated sector first (it acts from the "left"), the transversal field begins by changing the unfrozen region and the less correlated sector (it acts from the "right"). We will further discuss the consequences of  our second main result, Theorem~\ref{thm:hierar}, in the next subsection and spell out its proof only in Section~\ref{sec:proof} below.

\subsection{Instability of the glass phase in the QGREM with longitudinal hierarchical field}\label{sec:atl}

 If $A = \bar{A}$, i.e.~$A$ is a concave function, $\varphi^{(y,1)}$ is a just a translation of $\varphi^{(0,1)}.=: \varphi $ such that
	\begin{equation}\label{eq:hierar2} \Phi(\beta,\mathfrak{b},h) =  \sup_{0 \leq y \leq z \leq 1}\left[ \beta \eta(y) + \int_{y}^{z} \varphi(\beta, x) \, dx + (1-z) \ee[\ln 2 \cosh(\beta \mathfrak{b})] \right], \end{equation}
	with 
	\[
	\varphi(\beta, x) = \beta \sqrt{(2 \ln2)\ \bar{a}(x)}  \mathbbm{1}_{x < x(\beta)} + 
\left(\frac{\beta^2}{2} \bar{a}(x) + \ln2\right)  \mathbbm{1}_{x \geq x(\beta)} , \quad  x(\beta)\coloneqq \sup \left\{x \ | \ \bar{a}(x) > (2 \ln2)/ \beta^2 \right\} .
\] 
On the other hand, if $A$ is not concave (which is always the case if $A$ is a step function) the behaviour of $\varphi^{(y,1)}$ is more subtle as one has to take into account that the slope of the concave hull's linear segments will change as $y$ increases. In particular, \eqref{eq:hierar2} does not necessarily hold true. In contrast to a transversal field, a hierarchical field might lead to a change of the determining concave hull. As discussed in~\cite{DG86a} this would happen for a hierarchical caricature of a $ p $-spin glass with $ p > 2 $.
	
For an explicit prediction on the de Almeida-Thouless (AT) line we will now focus on the case that $A = \bar{A}$ is continuously differentiable with derivative $ \bar{a} $. 
Then for any hierarchical field with an overlap function $\eta(\cdot) = h v(\cdot)$ with $h \geq 0$ and $v \geq 0$ an increasing function, the supremum in \eqref{eq:hierar2} 
is attained for  fixed $\beta \geq 0$ at some $y(\beta,h)$ which is an increasing function of $h$. Since the critical temperature $T_c = \beta_c^{-1}$ only depends on $\bar{a}(y(\beta,h))$, it is thus a decreasing function of $ h $ and not increasing as in the QREM. \\

To be more specific, let us focus on the case of the hierarchical magnetic field $ \eta = h \gamma $ of strength $h > 0 $. We will proceed step by step, first discussing the limiting cases.\\

\noindent
	\textbf{Vanishing transversal field $\mathfrak{b} = 0 $:} \\[1ex]
In this case, a straightforward differentiation shows that the supremum in~\eqref{eq:hierar2} is attained at $ z = 0 $ and $ y= y(\beta, h) \in (0,1) $, which for fixed $ \beta > 0 $ and $ h > 0 $ is the unique solution of the equation
\begin{equation}\label{eq:y}
 y   = k\left( \frac{\varphi\left(\beta, y\right)}{\beta h } \right)  ,
\end{equation} 
where $ k: [0,\infty) \to (0,1] $ is the inverse function of the derivative $ \gamma' :(0,1]\to [0,\infty) $ of $ \gamma $. The uniqueness of the solution is most easily seen using the explicit form
\[
k(x) = \begin{cases} 1 , & x= 0, \\  \frac{1}{x} \tanh\frac{\ln 2}{x}  - \frac{1}{\ln 2} \ln \cosh\frac{\ln2}{x} , &  x>0 . \end{cases}
\]
from which we conclude the fact that $ k $ is continuous and monotone decreasing. More precisely, since $y \mapsto \varphi( \beta, y) $ is continuous and monotone decreasing as well with limiting values $\varphi(\beta, 0) \geq  \varphi(\beta, 1) = \beta^2 \bar{a}(1)/ 2 + \ln 2 > 0 $, the solution to~\eqref{eq:y} exists and is unique.

A low-temperature glass phase occurs in this case if and only if  $ y(\beta,h) <  x(\beta) $. 
Clearly, this is only possible in case $ x(\beta) > 0 $, i.e.\  for temperatures below the critical temperature at $ h= 0 $, whose inverse is given by
\[
\beta_c :=  \sqrt{\frac{2 \ln 2}{  \bar{a}(0) }} ,
\] 
Since $[\beta_c,\infty) \ni \beta \mapsto x(\beta) $ is monotone increasing and right-continuous and $ \varphi(\beta,x(\beta)) = 2 \ln 2 $, the inverse critical temperature at $h > 0 $ is then well defined thought the requirement
\begin{equation}\label{eq:defcritbeta} 
  \beta_c(h) \coloneqq  \inf \left\{\beta \ | \ x(\beta) > k\left(2 \ln2/ (\beta h)\right) \right\} . 
\end{equation}
The function $ h \mapsto \beta_c(h) $ is referred to as the AT line. 
We record some elementary properties of the AT line and also of the solution of~\eqref{eq:y}  for future purposes in the following proposition.
Of particular interest is the critical exponent of the AT line $T_c(h) = \beta_c(h)^{-1} $ near $ h = 0 $. It is determined by the asymptotic behaviour of $ \bar{a}(x)  $ near $ x = 0 $.
To facilitate notation, we write $  x(t) \propto y(t) \; ( t \to t_0) $ if and only if $ \lim_{t\to t_0 } \frac{x(t)}{y(t) } \in (0,\infty) $ exists. For the determination of the critical exponent, we add the following assumption, which may be satisfied or not.
\begin{assumption}\label{ass:a}
For  $ \alpha > 0 $: \quad 
$
\bar{a}(0) - \bar{a}(x) \propto  x^\alpha  \quad ( x \downarrow 0) 
$.
\end{assumption}
E.g.\ in the SK-caricature case $ A = \gamma^2 $, we have $ \bar{a}(0) = 2 \ln 2 $, which yields the correct critical temperature $ \beta_c = 1 $ of the SK-model, and $ \alpha = 1 $. As is spelled out in~\eqref{eq:critexp}, this leads to the critical exponent $2 $  of the AT-line for small transversal fields. This differs from the known asymptotics $ T_c - T_c(h) \propto h^{2/3} \quad ( h \downarrow 0 ) $ of the AT-line  in the original SK-model as already noted  in~\cite{DG86a}. 
\begin{proposition}\label{prop:bsk}
Suppose that $A = \bar{A}$ is continuously differentiable with derivative $ \bar{a} $.
\begin{enumerate}
\item The inverse critical temperature $\beta_c(h) $ is monotone increasing in $ h $. Its limiting values are $ \lim_{h\downarrow 0} \beta_c(h) = \beta_c$ and
\[  \lim_{h\to \infty} \beta_c(h) = \begin{cases} \infty& \mbox{if}\;   \bar{a}(1) = 0 , \\  \frac{2 \ln 2}{  \bar{a}(1)} &  \mbox{if}\;  \bar{a}(1) > 0. \end{cases} \]
In the situation of Assumption~\ref{ass:a} the critical temperature satisfies:
\begin{equation}\label{eq:critexp}
T_c - T_c(h) \propto h^{2\alpha} \quad ( h \downarrow 0 ).
\end{equation}
 \item 
For any $ \beta \in (0,\infty) $ and $ h> 0 $ the unique solution $ y(\beta,h) $ of~\eqref{eq:y} enjoys the following properties: 
\begin{enumerate}
\item For fixed $ \beta  \in (0,\infty) $ the function $  (0,\infty) \ni  h \mapsto y(\beta, h) $ is continuous and increasing in $ h $ for any $ \beta > 0 $ with limiting values $ \lim_{h\downarrow 0} y(\beta,h) = 0 $ and $\lim_{h\to \infty} y(\beta,h) = 1$. Moreover, 
\begin{equation}\label{asympy}
 y(\beta,h) =   \mathcal{O}(h^3 ) + h^2 \times \begin{cases}  \frac{ \beta^2}{ 2\ln 2}  (1+ \beta^2/\beta_c^2)^{-2}  , &\quad \beta < \beta_c \\[.5ex]
		 \frac{\beta_c^2}{8 \ln 2 } , & \quad \beta > \beta_c 
	\end{cases} 
\end{equation}
	for small $ h $. 
\item The function $  (0,\infty) \ni  h \mapsto  \varphi\left(\beta, y(\beta,h) \right) $ is continuous and decreasing. Moreover, at any  $ \beta > 0 $ its limiting values is $ \lim_{h\downarrow 0} \varphi\left(\beta, y(\beta,h) \right)  = \varphi\left(\beta, 0\right) $. 
\end{enumerate}
\end{enumerate}
\end{proposition}
The proof of this proposition consists again of multiple lengthy, but elementary computations, which are sketched in the appendix.\\

\noindent
	\textbf{Vanishing hierarchical longitudinal field $h = 0 $:} \\[1ex]
It was shown in Corollary~1.5 of~\cite{MW20c} that in case $ h = 0 $ and a constant transversal field $\mathfrak{b} = \Gamma$ of strength $\Gamma > 0$ the supremum in~\eqref{eq:hierar2} is attained at $ y = 0 $ and $ z = z(\beta, \Gamma) \in [0,1] $ given by
\begin{equation}\label{def:z}
z(\beta, \Gamma) := \begin{cases}  1 & p(\beta \Gamma) \leq s(\beta) = \varphi(\beta,1)  \\
    g_{\beta}(p(\beta \Gamma)) & s(\beta) < p(\beta \Gamma) < t(\beta) := \varphi(\beta,0) \\
    0 & t(\beta)\geq  p(\beta \Gamma) .
    \end{cases}
\end{equation}
Here  $g(\beta,\cdot): [s(\beta),t(\beta)] \to  [0,1]  $ is the generalized inverse of $\varphi(\beta,\cdot) $, which maximizes $ z(\beta,\Gamma) $ and 
\[ p(\beta \Gamma):= \ln 2 \cosh(\beta \Gamma) , \]  
is the pressure of a pure quantum paramagnet.
As a consequence, the pressure $ \Phi(\beta,\Gamma,0) $ has a magnetic transition at 
\[ \Gamma_c(\beta,0) :=\frac{1}{\beta} \arcosh\left(\tfrac{1}{2} e^{t(\beta)} \right) \]
and possibly a second magnetic transition at $
\Gamma_c^{(1)}(\beta) := \frac{1}{\beta} \arcosh\left(\tfrac{1}{2} e^{s(\beta)} \right) $ depending on whether $ \bar{a}(1) > 0 $ or equivalently $ s(\beta) > \ln 2 $ or not.
In the regime $ \Gamma < \Gamma_c(\beta,0) $ a glass transition occurs at fixed inverse temperature~$ \beta_c $. 

In case of the SK-caricature for which $ \bar{a}(1) = 0 $, neither the value of the location of the quantum phase transition at zero temperature, $ \lim_{\beta \to \infty} \Gamma_c(\beta,0) = \sqrt{(2\ln 2) \bar{a}(0) } = 2 \ln 2  \approx 1.38\dots $ agrees with the perturbative or numerical prediction of approximately $ 1.51 $ in~\cite{YI86,You17}, nor does the behaviour of $ \Gamma_c(T^{-1},0) $ near $ T= 0 $ agree with the  $  T^2 $-scaling predicted in~\cite{HM93}. Presumably, this is a defect of the hierarchical implementation of glass. \\

\noindent
	\textbf{Constant longitudinal and transversal field:} \\[1ex]
To determine the pressure $ \Phi(\beta,\Gamma,h) $ in the general case of a constant transversal and longitudinal field $ \Gamma, h > 0 $, we also need to discuss the behaviour of the variational expression~\eqref{eq:hierar2}  at the diagonal $ y = z $, which corresponds to the situation without a GREM.  In this case,  the supremum is attained at 
\begin{equation}\label{eq:s_0}
 \sigma(\beta,\Gamma,h)   :=  k\left(  \frac{p(\beta \Gamma)}{\beta h } \right)  ,
\end{equation}
Note that the condition $ p(\beta\Gamma) < \varphi(\beta,y(\beta,h)) $ ensures $ y(\beta,h) < z(\beta,h) $ by the strict monotonicity of $ g_\beta $. 
These findings then yield to the following explicit expression for the pressure in the general case. 
\begin{corollary}\label{cor:hG}
Suppose that $A = \bar{A}$ is continuously differentiable. For the constant transversal field of strength $ \Gamma > 0 $ and the hierarchical magnetic field  $h(\pmb{\sigma}) = N h \gamma(q(\pmb{\sigma},\pmb{\sigma}^0) )$ of strength $ h > 0 $ the pressure is almost surely
\[
\Phi(\beta,\Gamma,h) = \begin{cases} \displaystyle \beta h \gamma\left( y(\beta,h) \right) + \int_{ y(\beta,h)}^{z(\beta, \Gamma)} \varphi(\beta,x) dx + \left(1-z(\beta, \Gamma)\right) p(\beta\Gamma) , &  p(\beta\Gamma) < \varphi(\beta,y(\beta,h)) \\ 
 \beta h \gamma\left( \sigma(\beta,\Gamma,h)\right)  + \left(1-\sigma(\beta,\Gamma,h) \right) p(\beta\Gamma) , &  p(\beta\Gamma) \geq \varphi(\beta,y(\beta,h)) ,
 \end{cases}
\]
 where $y(\beta,h ) $, $ z(\beta,\Gamma) $ and $ \sigma(\beta,\Gamma,h) $ are specified in \eqref{eq:y}, \eqref{def:z} and \eqref{eq:s_0} respectively. 
\end{corollary}

Let us now discuss the physical significance of this formula. 
In case $ h > 0 $ the pressure in Corollary~\ref{cor:hG} changes its nature at $  \varphi(\beta,z(\beta,\Gamma))  = p(\beta\Gamma) = \varphi(\beta,y(\beta,h)) $, i.e.\ at
\[ 
\Gamma_c(\beta,h) := \frac{1}{\beta} \arcosh\left(\tfrac{1}{2} e^{\varphi(\beta,y(\beta,h)) } \right) . 
\]
By strict monotonicity of $ p $,  the condition $ \Gamma < \Gamma_c(\beta,h) $ is equivalent to $ p(\beta\Gamma)  < \varphi(\beta, y(\beta,h) )$ and hence $y(\beta,h) \leq  z(\beta,\Gamma) $.

The magnetization in the transversal direction 
\[
m_x(\beta,\Gamma,h) := \frac{1}{\beta} \frac{\partial}{\partial \Gamma} \Phi(\beta,\Gamma,h) = \begin{cases} 
(1-z(\beta,\Gamma) ) \tanh\beta\Gamma , & p(\beta\Gamma) <  \varphi(\beta,y(\beta,h)) ,\\
 (1-\sigma(\beta,\Gamma,h) ) \tanh\beta\Gamma , & p(\beta\Gamma) \geq  \varphi(\beta,y(\beta,h)) ,
\end{cases} 
\]
changes continuously through the transition line $ \Gamma = \Gamma_c(\beta,h) $. Only its second derivative is generally discontinuous. Note that the magnetization in $ x $-direction neither attains its maximum value $  \tanh(\beta\Gamma)$ of the pure quantum paramagnetic phase in the regime $ \Gamma > \Gamma_c(\beta,h)   $ nor does it vanish for $  \Gamma <  \Gamma_c(\beta,h)  $. Similarly as in the case $ h = 0 $ covered in~\cite{MW20c}, the transversal magnetization vanishes only at $
\Gamma_c^{(1)}(\beta)  $, which is equal to zero in case $ \bar{a}(1) = 0 $. 
The critical magnetic field is continuous in $ h $, and one recovers the limiting value 
$\lim_{h\downarrow 0} \Gamma_c(\beta,h) = \Gamma_c(\beta,0)$ for any $ \beta \in (0,\infty) $.
A straightforward Taylor expansion and~\eqref{asympy} imply that in the situation of 
Assumption~\ref{ass:a}:
\begin{equation}\label{eq:critG}
\Gamma_c(\beta,0) -\Gamma_c(\beta,h) \sim h^{2\alpha}  \quad ( h \downarrow 0 ).
\end{equation}
In fact, this even holds in the zero temperature limit $ \beta \to \infty $, i.e for the so called Quantum AT line which is plotted in Figure~\ref{fig:QAT}. 
\begin{figure}[ht]
	\begin{center}
		\includegraphics[width=.42\textwidth]{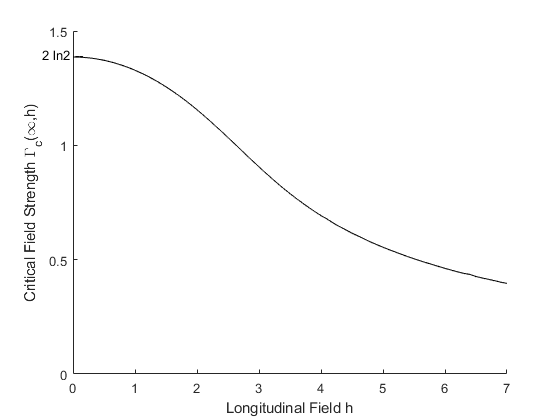}
	
		\caption{Plot of the Quantum AT line, i.e.~the dependence of the critical transversal field $\Gamma_c(\beta,h)$ on the longitudinal field $ h $ for zero temperature, $ \beta = \infty $. }\label{fig:QAT}
	\end{center}
\end{figure}

A low-temperature glass phase occurs if and only if 
\[
 y(\beta,h) < \min\left\{ x(\beta), z(\beta,\Gamma) \right\} .
\]
Clearly, this is only possible if two conditions are satisfied simultaneously:
\begin{enumerate}
\item $ z(\beta,\Gamma) >  y(\beta,h)  $, i.e.~for transversal fields $ \Gamma < \Gamma_c(\beta,h) $.
From the monotonicity of $ h \mapsto \varphi(\beta,y(\beta,h)) $,  we conclude, 
 $ \Gamma_c(\beta,h) \leq \Gamma_c(\beta, 0) $
for any $ \beta , h > 0 $. 
\item $  x(\beta) >   y(\beta,h) $, i.e.\  for $ \beta > \beta_c(h) $ given by~\eqref{eq:defcritbeta}, which we already identified as a monotone increasing function of $ h $.
\end{enumerate}
We thus conclude, that the presence of the transversal field $ h > 0 $ shrinks the spin glass' low-temperature phase. Qualitatively this behaviour is in accordance with the numerical findings in case of the Quantum SK-model \cite{You17}.  However, as already noted in~\cite{DG86a} in the classical case $ \Gamma = 0 $, the critical exponents do not agree.
Figure~\ref{fig:phaseGREM} plots the temperature-transversal field  phase diagram for different values of $ h $ in case $ A = \bar{A} $ and $ \bar{a}(1) = 0 $.
\begin{figure}[ht]
	\begin{center}
		\includegraphics[width=.43\textwidth]{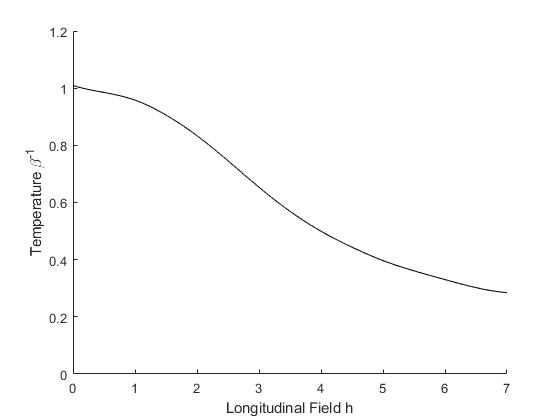}
		\includegraphics[width=.53\textwidth]{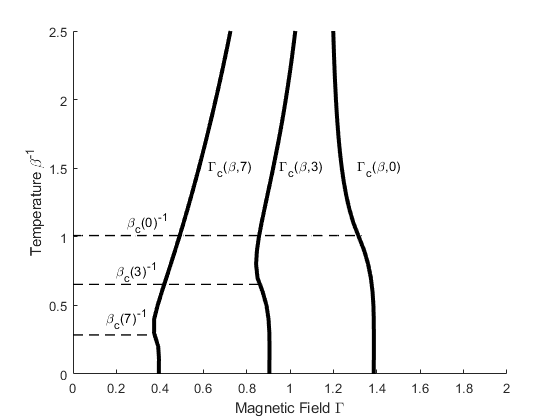}
		\caption{On the left is a plot of the critical temperature $ \beta_c(h) $ as a function of the longitudinal field. On the right figure is the $T-\Gamma$ phase diagram with the critical magnetic field $\Gamma_c(\beta,\Gamma)$ as well as the critical temperature $\beta_c(h)^{-1} $ evaluated at $h = 0,3,7$. }\label{fig:phaseGREM}
	\end{center}
\end{figure}

We finally close this section by pointing out that the expression for the pressure in case $ p(\beta\Gamma) \geq  \varphi(\beta,y(\beta,h)) $ agrees with that of the hierarchical field $ h $ plus a constant transversal field $ \Gamma $. It should be compared to the exact solution $ p(\beta\sqrt{h^2 +\Gamma^2})  $ without the hierarchical implementation of the longitudinal field and agrees qualitatively. 

The magnetization in the longitudinal direction is given by
\[
m_z(\beta,\Gamma,h) := \frac{1}{\beta} \frac{\partial}{\partial h} \Phi(\beta,\Gamma,h) = \begin{cases} \gamma\left( y(\beta,h)\right) ,
 & p(\beta\Gamma) <  \varphi(\beta,y(\beta,h)) \\
 \gamma\left(\sigma(\beta,\Gamma,h) )\right) , & p(\beta\Gamma) \geq  \varphi(\beta,y(\beta,h)) 
\end{cases} 
\]
and varies continuously through both the glass and the magnetic transitions.

\section{Proof of Theorem~\ref{thm:hierar}}\label{sec:proof}

Let us first remark 
that the  last equality in \eqref{eq:hierar} already follows from results in \cite{MW20c}. Indeed, fix any $y \in [0,1)$ and consider the Hamiltonian 
 \[ H^{(y)} := U^{(y,1)} - B^{2,y} \]
 on the reduced Hilbert space $\ell^2(\mathcal{Q}_{\lceil (1-y)N \rceil})$, where $ U^{(y,1)} $ is the cut GREM corresponding to $ A^{(y,1)} $ and $B^{(2,y)}$ denotes the cut transversal field acting only on spins in $ \mathcal{Q}_{\lceil (1-y)N \rceil} $
 \begin{equation}\label{eq:cutB}  
 	B^{(2,y)} \coloneqq  \sum_{i=\lceil y N \rceil+1}^{N} b_i \pmb{s}_i   ,
 \end{equation}
 and we set $B^{(1,y)} \coloneqq B-B^{(2,y)}$.
   Then, Theorem 2.8 in \cite{MW20c} implies 
 \[ \lim_{N \to \infty} \frac1N \ln \tr e^{-\beta H^{(y)}} = \sup_{y \leq z \leq 1}\left[ \int_{0}^{z-y} \varphi^{(y,z)}(x) \, dx + (1-z) \ee[\ln 2 \cosh(\beta \mathfrak{b})] \right] ,  \]
 whereas an application of Theorem 1.4 in \cite{MW20c}  yields
 \[ \lim_{N \to \infty} \frac1N \ln \tr e^{-\beta H^{(y)}} = \sup_{y \leq z \leq 1}\left[ \int_{0}^{z-y} \varphi^{(y,1)}(x) \, dx + (1-z) \ee[\ln 2 \cosh(\beta \mathfrak{b})] \right] .  \]
In both cases, the supremum is taken over $ z \in [y,1] $ at fixed $ y $, which proves  the second equality in \eqref{eq:hierar}. 
We now spell out the proof of the first equality in \eqref{eq:hierar}.

\begin{proof}[Proof of Theorem~\ref{thm:hierar}]
	We will proceed in three steps. \\[1ex]
	\noindent
	\textbf{Step 1: Reduction to step functions}  \nopagebreak \\[1ex]
	We claim that it is enough to show Theorem~\ref{thm:hierar}
	for step functions $\eta$. This follows if we can prove that the left and right side of \eqref{eq:hierar} are continuous with respect to $\eta$ in the uniform norm. This is, however, trivial for the right side, and a simple operator norm bound implies 
	for two hierarchical fields $h,h'$ with overlap functions $\eta,\eta'$,
	\[\frac1N \left|\ln \tr e^{-\beta(U - h - B)} -\ln \tr e^{-\beta(U -h'-B)}\right| \leq \beta \|\eta-\eta'\|_{\infty}.   \]
	 From now on, we will therefore only consider step functions $\eta$, i.e.\ we assume that there exist points $0 = q_0 < q_1 < \cdots q_m = 1$ 
	and real numbers $\eta_1,\ldots,\eta_m$ such that $\eta(x) = \eta_k$ for $q_{k-1} \leq x < q_k$ and $\eta(1) = \eta_m$. The points $q_k$ define blocks of spin vectors $\pmb{\sigma}_k \in \mathcal{Q}_{\lceil q_k N \rceil - \lceil q_{k-1} N \rceil}$, and we will write $\pmb{\sigma} = \pmb{\sigma}_1 \pmb{\sigma}_2 \cdots \pmb{\sigma}_m$. Moreover, it is convenient to introduce for $k = 1,\ldots m$ the projections $P_k$ and $p_k$:
	\[ P_k \pmb{\sigma} = P_k \pmb{\sigma}_1 \pmb{\sigma}_2 \cdots \pmb{\sigma}_m \coloneqq \pmb{\sigma}_1\cdots \pmb{\sigma}_k,  \qquad  p_k \pmb{\sigma} = p_k \pmb{\sigma}_1 \pmb{\sigma}_2 \cdots \pmb{\sigma}_m \coloneqq \pmb{\sigma}_k.  \]
	Moreover, we set $P_0 \pmb{\sigma} = p_0 \pmb{\sigma} \coloneqq 1$.
	 Finally, we note that due to the fact that $ \eta $ only takes finitely many values, we may restrict the variational expression \eqref{eq:hierar} to the maximum over points $y= q_k$: 
	\[ \begin{split} &\sup_{0 \leq y \leq z \leq 1}\left[  \beta \eta(y) + \int_{0}^{z-y} \varphi^{(y,1)}(x) \, dx + (1-z) \ee[\ln 2 \cosh(\beta \mathfrak{b})] \right] \\ =& \max_{k=0,\ldots,m-1} \sup_{q_k \leq z \leq 1} \left[ \beta \eta_{k+1}  + \int_{0}^{z-q_k} \varphi^{(y,1)}(x) \, dx + (1-z) \ee[\ln 2 \cosh(\beta \mathfrak{b})] \right] .\end{split} \]
		
	\noindent
	\textbf{Step 2: Lower bound} \\[1ex]
	Our lower bound on the presssure is based on Gibbs' variational principle~\cite{Rob67}. 
	We pick some $k \in \{1,\ldots,m \}$ and consider 	on the subspace $\ell^2(\mathcal{Q}_{N- \lceil q_k N \rceil })$ the Hamiltonian:
	\begin{equation}\label{def:HUk}
	H^{(k)} := U^{(k)}  - B^{2,q_k}, \qquad U^{(k)}(\pmb{\sigma}_{k+1}\cdots \pmb{\sigma}_m) := U((P_k \pmb{\sigma}^0)\pmb{\sigma}_{k+1}\cdots \pmb{\sigma}_m) . \end{equation}
 We denote by $\tilde{\rho}_{k,\beta}$ the corresponding Gibbs state at inverse temperature $ \beta $. The density matrix $\tilde{\rho}_{k,\beta}$ has the extension $\rho_{k,\beta} := |  P_k \pmb{\sigma}^0 \rangle \langle  P_k \pmb{\sigma}^0 | \otimes  \tilde{\rho}_{k,\beta} $ to the full space $\ell^2(\mathcal{Q}_{N}) = \ell^2(\mathcal{Q}_{\lceil q_k N \rceil }) \otimes  \ell^2(\mathcal{Q}_{N- \lceil q_k N \rceil }) $, whose matrix elements are given by 
	\[ \langle \pmb{\sigma}| \rho_{k,\beta} | \pmb{\sigma}^\prime \rangle \coloneqq \begin{cases}  \langle \pmb{\sigma}_{k+1} \cdots \pmb{\sigma}_m | \tilde{\rho}_{k,\beta} | \pmb{\sigma}_{k+1}^\prime\cdots \pmb{\sigma}_{m}^\prime \rangle & \text{if } P_k \pmb{\sigma} = P_k \pmb{\sigma}^0 = P_k \pmb{\sigma}^\prime \\
	0 & \text{ else }.
	\end{cases}    \]
	By Gibbs' variational principle, we have 
	\[ \frac1N \ln \tr e^{-\beta(U- h-B)} \geq \frac{\beta}{N} \tr[\rho_{k,\beta}(B^{1,q_k}+h + U^{(k)}  - U)] + \frac1N \ln \tr_{|\ell^2(\mathcal{Q}_{N- \lceil q_k N \rceil }) } e^{-\beta H^{(k)}}. \]
 Since  the trial density matrix $\rho_{k,\beta}$ is diagonal with respect to $\pmb{\sigma}_1 \cdots \pmb{\sigma}_k$ and fixes the first variables to $ P_k \pmb{\sigma}^0 $, we have 
  \[ \tr[\rho_{k,\beta}B^{1,q_k}] = 0 =  \tr\big[\rho_{k,\beta} (U^{(k)}  - U)\big] .\] 
  Thus, it remains to show the almost sure identities
  \begin{equation}\label{eq:low1}
  \lim_{N \to \infty} \frac1N \tr[\rho_{k,\beta} h] = \eta_{k+1},
  \end{equation}
  and 
  \begin{equation}\label{eq:low2}
  \lim_{N \to \infty} \frac1N \ln \tr_{|\ell^2(\mathcal{Q}_{N- \lceil q_k N \rceil }) } e^{-\beta H_k} = \sup_{q_k \leq z \leq 1} \left[\int_{0}^{z-q_k} \varphi^{(y,1)}(x) \, dx + (1-z) \ee[\ln 2 \cosh(\beta \mathfrak{b})] \right] .
  \end{equation}
  	
\noindent
  \textit{Step 2.1: Proof of~\eqref{eq:low1}:} \quad
  Using $ h(\pmb{\sigma}) = N \eta(q_N(\pmb{\sigma},\pmb{\sigma}^0)) $ we compute the trace in the $ z$-basis:
  \[ \begin{split} \frac{1}{N} \tr[\rho_{k,\beta}h ] &= \sum_{l=0}^{m-1} \left( \eta_{l+1} \sum_{\pmb{\sigma}: P_{l}\pmb{\sigma}^0 = P_{l} \pmb{\sigma}, P_{l+1}\pmb{\sigma}^0 \neq P_{l+1} \pmb{\sigma}} \langle \pmb{\sigma}|\rho_{k,\beta} | \pmb{\sigma} \rangle \right) +  \eta_m \langle \pmb{\sigma}^0| \rho_{k,\beta} | \pmb{\sigma}^0 \rangle \\
  &= \sum_{l=k}^{m-1} \left(  \eta_{l+1} \sum_{\pmb{\sigma}: P_{l}\pmb{\sigma}^0 = P_{l} \pmb{\sigma}, P_{l+1}\pmb{\sigma}^0 \neq P_{l+1} \pmb{\sigma}} \langle \pmb{\sigma}|\rho_{k,\beta} | \pmb{\sigma} \rangle \right) + \eta_m \langle \pmb{\sigma}^0| \rho_{k,\beta} | \pmb{\sigma}^0 \rangle,
   \end{split}  \] 
   where the second equality is due to the construction of $\rho_{k,\beta}$. Since   $\rho_{k,\beta}$ has unit trace, 
   \begin{equation}\label{eq:tr1}
   1 =  \sum_{\pmb{\sigma}:  P_{k}\pmb{\sigma}^0 = P_{k} \pmb{\sigma}} \langle \pmb{\sigma}|\rho_{k,\beta} | \pmb{\sigma} \rangle , 
   \end{equation} 
   and is non-negative, we may estimate both from above and below:
   \[ \left| \frac1N \tr[\rho_{k,\beta} h ] - \eta_{k+1} \right| \leq \|\eta\|_{\infty} \sum_{\pmb{\sigma}:  P_{k+1}\pmb{\sigma}^0 = P_{k+1} \pmb{\sigma}} \langle \pmb{\sigma}|\rho_{k,\beta} | \pmb{\sigma} \rangle.  \]
   We further deduce from the spin-flip  covariance of $ H^{(k)} $ that     for any $\pmb{\sigma},\pmb{\sigma}'$ with $P_k \pmb{\sigma} = P_k \pmb{\sigma}' = P_k \pmb{\sigma}^0$:
   \[ \ee[\langle \pmb{\sigma}|\rho_{k,\beta} | \pmb{\sigma} \rangle] = \ee[ \langle \pmb{\sigma}'|\rho_{k,\beta} | \pmb{\sigma}' \rangle] . \]
Consequently, using the normalisation~\eqref{eq:tr1} and counting the number of configurations, we have
    \[ \ee\left[ \sum_{\pmb{\sigma}:  P_{k+1}\pmb{\sigma}^0 = P_{k+1} \pmb{\sigma}} \langle \pmb{\sigma}|\rho_{k,\beta} | \pmb{\sigma} \rangle \right] = \frac{2^{N(1-q_{k+1})}}{ 2^{N(1-q_{k})} } = 2^{-N(q_{k+1} -q_k)}.  \]
    By a Borel-Cantelli argument, we thus arrive at the almost sure convergence 
    \[ \lim_{N \to \infty} \left| \frac1N \tr[\rho_{k,\beta}h ] - \eta_{k+1} \right| = 0.   \]
    	
\noindent
    \textit{Step 2.2: Proof of~\eqref{eq:low2}: } \quad
    We may rewrite the restricted process (in distributional sense)
    \[U((P_k \pmb{\sigma}^0)\pmb{\sigma}_{k+1}\cdots \pmb{\sigma}_m)
      = U^\prime(\pmb{\sigma}_{k+1}\cdots \pmb{\sigma}_m) + \sqrt{N  A(q_k)} \ Y, \]
     where $U^\prime(\pmb{\sigma}_{k+1}\cdots \pmb{\sigma}_m)$ is a GREM process on $ \mathcal{Q}_{N- \lceil q_k N \rceil } $ with (non-normalized) distribution function $A^{(q_k,1)}$ and $Y$ is a standard Gaussian variable which is independent of $U^\prime$. This  distributional equality relies on the fact that centered Gaussian processes are uniquely determined by their covariance function. Of course, $Y$ does not contribute to the limit of the pressure, 
     \[ \lim_{N \to \infty} \frac1N \ln \tr_{|\ell^2(\mathcal{Q}_{N- \lceil q_k N \rceil }) } e^{-\beta H^{(k)}} = \lim_{N \to \infty} \frac1N \ln \tr_{|\ell^2(\mathcal{Q}_{N- \lceil q_k N \rceil }) } e^{-\beta(U^\prime - B^{(2,q_k)})},    \]
     provided that the limit on the right side exists. This is warranted by Theorem 1.4 in \cite{MW20c}, which almost surely yields 
     \begin{equation}\label{eq:Step2.2}
     \lim_{N \to \infty} \frac1N \ln\tr_{|\ell^2(\mathcal{Q}_{N- \lceil q_k N \rceil }) } e^{-\beta(U^\prime - B^{(2,q_k)})} =  \sup_{q_k \leq z \leq 1}\left[  \int_{0}^{z-q_k} \varphi^{(y,1)}(x) \, dx + (1-z) \ee[\ln 2 \cosh(\beta \mathfrak{b})] \right] . 
     \end{equation}
     
\bigskip
     \noindent
     \textbf{Step 3: Upper bound}  \\[1ex]
     The method is similar in spirit to the application of the peeling principle presented in \cite{MW20c}. However, we need to cut the transversal field $B$ in a different manner which suits the hierarchical field $h$.\\
    	
\noindent   \textit{Step 3.1: Truncating the transversal field $B$} \\[1ex]
     We define the partial fields 
     \[ B_k \coloneqq B^{(1,q_k-1)} - B^{(1,q_{k})} = \mkern-5mu\sum_{i=\lceil q_{k-1} N \rceil+1}^{\lceil q_{k} N \rceil} \mkern-10mu b_i \pmb{s}_i     \]
     where we set $B^{(1,q_0)} = 0$. Hence $B_k$ only acts on $\pmb{\sigma}_k$. We also define the restriction $B^\prime_k$ of $B_k$ to the complement of $(p_k \pmb{\sigma}^0)$:
     \[ B_k - B^\prime_k := \mathbbm{1} \otimes \sum_{\lceil q_{k-1} N \rceil < j \leq \lceil q_{k} N \rceil} b_j  \left( |p_k(F_j\pmb{\sigma}^0) \rangle \langle p_k\pmb{\sigma}^0 | + |p_k\pmb{\sigma}^0 \rangle \langle p_k(F_j\pmb{\sigma}^0)| \right) \otimes \mathbbm{1}.    \]
     Here, the first identity acts on $\pmb{\sigma}_{1}\cdots \pmb{\sigma}_{k-1}$,the last identity on $\pmb{\sigma}_{k+1}\cdots \pmb{\sigma}_m$ and $F_j$ denotes the $j$th flip operator (see~\eqref{def:B}).  We denote by $B^\prime$ the total truncated  transversal field,
     \[B^\prime = \sum_{k=1}^{m} B_k^\prime  \]
   By  the triangle inequality and a Frobenius norm estimate we have
     \[ \|B - B^\prime\| \leq \sum_{k=1}^{m} \|B_k - B_k^\prime\|  \leq m \sqrt{2  \sum_{i=1}^{N} |b_i|^2} = o(N).  \]
     Note that the $L^1$-property of the random variable $\mathfrak{b}$ and Lemma A.2 in \cite{MW20c} ensure that the right side is indeed of order  $o(N)$.\\[1ex]
     
     \noindent
      \textit{Step 3.2: Finishing the proof:}\quad
     Using a trivial norm bound, we estimate
      \[ \begin{split} &  e^{-\beta \|B - B^\prime\| }\ \tr e^{-\beta(U- h - B)} \leq 
      \tr e^{-\beta(U- h - B^\prime)} \\
      & = \sum_{k=0}^{m-1} \left( e^{-\beta N \eta_{k+1}} \sum_{\pmb{\sigma}:  P_{k}\pmb{\sigma}^0 = P_{k} \pmb{\sigma}, P_{k+1}\pmb{\sigma}^0 \neq P_{k+1} \pmb{\sigma} } \langle \pmb{\sigma} | e^{-\beta(U - B^\prime)} | \pmb{\sigma} \rangle \right) 
        + e^{-\beta N \eta_m} \langle \pmb{\sigma}^0 | e^{-\beta(U- B^\prime)} |\pmb{\sigma}^0 \rangle  \\
        & = \sum_{k=0}^{m-1} \left( e^{-\beta N \eta_{k+1}} \mkern-35mu \sum_{\pmb{\sigma}:  P_{k}\pmb{\sigma}^0 = P_{k} \pmb{\sigma}, P_{k+1}\pmb{\sigma}^0 \neq P_{k+1} \pmb{\sigma} } \mkern-30mu \langle \pmb{\sigma}_{k+1} \dots \pmb{\sigma_m} | e^{-\beta(U^{(k)} - B^{\prime,2,q_{k}})} | \pmb{\sigma}_{k+1} \dots \pmb{\sigma_m} \rangle \right) + e^{-\beta N \eta_m} e^{-\beta U(\pmb{\sigma}^0)}. \\
        	\end{split}     \]
        The first identity follows by an inclusion-exclusion type of summation over all spin configurations $\pmb{\sigma} \in\mathcal{Q}_N $ together with the fact that the hierarchical field $h$ commutes with $B^\prime$ (and clearly with $U$) and is constant on the respective spin configurations in the sum. The third line is a consequence of the fact that on the subspace generated by the elements $P_{k}\pmb{\sigma}^0 = P_{k} \pmb{\sigma}$, the magnetic field $ B^\prime $ operates only on the remaining spins $\pmb{\sigma}_{k+1} \cdots \pmb{\sigma}_m$ and evaluates the potential at $ U^{(k)} $,         
         see~\eqref{def:HUk}.  
         We now recall from Lemma~1 in  \cite{MW20c}  that the diagonal matrix elements $  \langle \pmb{\sigma} | e^{-\beta(U - B)} | \pmb{\sigma} \rangle $ only depend on the square of the variables $ b_i $, so that in the estimation of the trace we may aways assume without loss of generality that $ b_i \geq 0 $ and hence $ B $ as well as $ B^\prime $ have positive matrix elements in the spin configuration basis, which for $ b_i \geq 0 $ dominate each other and in particular
         \[ 
         0 \leq \langle \pmb{\sigma}_{k+1} \dots \pmb{\sigma_m} | e^{-\beta(U^{(k)} - B^{\prime(2,q_{k})})} | \pmb{\sigma}_{k+1} \dots \pmb{\sigma_m} \rangle  \leq \langle \pmb{\sigma}_{k+1} \dots \pmb{\sigma_m} | e^{-\beta(U^{(k)} - B^{(2,q_{k})})} | \pmb{\sigma}_{k+1} \dots \pmb{\sigma_m} \rangle .
         \]
This allows us to expand the summation over all matrix elements with $P_k \pmb{\sigma}^0 = P_k \pmb{\sigma}$, which leads to the upper bound
\[  e^{-\beta \|B - B^\prime\| }\ \tr e^{-\beta(U- h - B)} \leq   \sum_{k=0}^{m-1} e^{-\beta N \eta_{k+1}} \tr_{|\ell^2(\mathcal{Q}_{N- \lceil q_k N \rceil }) } e^{-\beta (U^{(k)}- B^{(2,q_{k})})} + e^{-\beta N \eta_m} e^{-\beta U(\pmb{\sigma}^0)}. 
\]
 Together with \eqref{eq:Step2.2} this finishes the proof of Theorem~\ref{thm:hierar}.
\end{proof}

\section{Proof of Theorem~\ref{thm:qcremh}}\label{sec:pfqcremh}
Based on the already established results and methods in \cite{AK14,AP19,MW20a,MW20c}, the proof of Theorem~\ref{thm:qcremh} is straightforward but  quite lengthy. Before we move on to the details, we outline our proof strategy which consists of three main steps:
\begin{itemize}
	\item[1.] First, we need to generalize the results in \cite{AK14,AP19} on the REM  and two-level GREM with a random magnetic field  to the general $n$-level GREM (see Theorem~\ref{prop:gremh} below). Following~\cite{AK14,AP19}  closely, the argument is based on a large-deviation principle for the entropy which transforms the computation of the limit to a linear optimisation problem with non-linear constraints. 
	\item[2.] Secondly, we extend the limit theorem for the classical GREM to the QGREM with a random longitudinal field (see Proposition~\ref{prop:qgremh} below). Using the peeling principle from \cite{MW20c}, the proof is quite easy. The only subtle point is to ensure that the structure of the concave hull in the variational principle is preserved. Here we use an argument which is very similar to the proof of \cite[Lemma 3.1]{MW20c}. 
	\item[3.] Finally, we use an interpolation and continuity argument to the lift the  $n$-level  QGREM result to the more general QGREM setting. We refer to the interpolation and concentration estimates in  \cite{MW20c} which are applicable
	here.
\end{itemize}

\subsection{The GREM with a random magnetic field}

The main aim of this subsection to prove the following Theorem~\ref{prop:gremh},   which extends the discussion of the two-level GREM in \cite{AP19} to the general $n$-level GREM. To this end, we will need to introduce some notation. Let $0 = x_0 < x_1 < x_2 < \cdots < x_n = 1$ be some points  $a_1,\ldots,a_n$ some nonnegative weights (we do not assume here that these weights add up to one). As in the proof of Theorem~\ref{thm:hierar}, we decompose the spin vector into blocks $\pmb{\sigma} = \pmb{\sigma}_1 \cdots \pmb{\sigma}_n$ according to the partition formed by the points $(x_k)$. 
The GREM process can be written as
\begin{equation}\label{eq:ugrem}
U(\pmb{\sigma}) = \sqrt{a_1 N} X_{\pmb{\sigma}_1} + \sqrt{a_2 N }X_{\pmb{\sigma}_1\pmb{\sigma}_2} + \cdots + \sqrt{a_n N }X_{\pmb{\sigma}_1\pmb{\sigma}_2\cdots \pmb{\sigma}_n}, 
\end{equation} 
where the appearing random variables $X_{\pmb{\sigma}_1}, X_{\pmb{\sigma}_1\pmb{\sigma}_2}, \ldots, X_{\pmb{\sigma}_1\pmb{\sigma}_2\cdots \pmb{\sigma}_n}$ are independent standard Gaussian variables.
Note that $U(\pmb{\sigma}) $ coincides with the GREM process with (non-normalized) distribution function $A$,
\[ A(x) = \sum_{k =1}^{n} a_k \mathbbm{1}_{[x_k,1]}(x). \]

The limit depends on the concave hull $\bar{A}$ of $A$ consisting of linear segments which are supported on a subset of points $0 = y_0 < y_1 < \cdots <y_m = 1$ where $A$ and $\bar{A}$ agree. It is convenient to further introduce the following quantities: the increments of the concave hull $\bar{a}_l \coloneqq A(y_l) - A(y_{l-1})$, the interval lengths $L_l \coloneqq y_l - y_{l-1}$ and the slopes $\gamma_l \coloneqq \bar{a}_l/L_l$. 

As our main result in this section, we show that the limit of the classical pressure
 $\Phi_N(\beta,\mathfrak{h}) = \Phi_N(\beta,\mathfrak{h},0)$ can then be expressed in terms of the partial pressures 
 \begin{equation}\label{eq:part}
 \varphi^{(l)}(\beta,\mathfrak{h})\coloneqq \begin{cases} 
  \frac{\bar{a}_l}{2} \beta^2 + L_l \ee[\ln 2 \cosh \beta \mathfrak{h}] & if \quad \beta \leq \beta_c^{(l)} \\
 	\beta(\bar{a}_l\beta_c^{(l)} + L_l \ee[\mathfrak{h} \tanh \beta_c^{(l)} \mathfrak{h}]) & if \quad \beta > \beta_c^{(l)} 
 \end{cases}
\end{equation}
where the critical temperatures $\beta_c^{(l)} = \beta_c^{(l)}(\mathfrak{h})$ are each the unique positive solution of the self-consistency equation 
\begin{equation}\label{eq:partcrit}
\frac{\gamma_l}{2}\beta_c^{(l)^2} =  \ln 2 + \ee[\ln \cosh \beta_c^{(l)} \mathfrak{h}] - \beta_c^{(l)}\ee[\mathfrak{h} \tanh \beta_c^{(l)} \mathfrak{h}]. 
\end{equation}
The following generalises results in~\cite{AK14,AP19}, which in turn built on~\cite{BK04a,BK08}.
\begin{theorem}\label{prop:gremh}
	Let $U(\pmb{\sigma}) $ be a GREM process as in \eqref{eq:ugrem},  $\beta \geq 0$ and $\mathfrak{h}$ an absolutely integrable random variable. Then, almost surely
	\begin{equation}\label{eq:gremh}
	\lim_{N \to \infty} \Phi_N(\beta,\mathfrak{h}) = \sum_{l=1}^{m} \varphi^{(l)}(\beta,\mathfrak{h}).
	\end{equation}
\end{theorem}

We stress that a random field does only change the partial pressures $ \varphi^{(l)}$ but not the number of terms in the right side. In particular, the limit remains to be a function of the concave hull $\bar{A}$ and not $A$ itself.\\

Our proof of Theorem~\ref{prop:gremh} follows the large-deviation approach in \cite{AK14,AP19}. We first need to understand the energy statistics of the random field. To this end, it is convenient to decompose the field $h(\pmb{\sigma})$ into blocks 
\[h_k(\pmb{\sigma}_k) \coloneqq \sum_{\lceil x_{k-1} N \rceil +1 \leq j \leq \lceil x_{k} N \rceil} h_j \sigma_j.     \]
We first study the occupation numbers 
\[ N(y_k) \coloneqq \left| \{ \pmb{\sigma}_k | \quad h_k(\pmb{\sigma}_k) \leq - N y_k    \} \right| . \]
With respect to the uniform distribution on spin configurations $ \pmb{\sigma}_k $, the random variables $  h_k(\pmb{\sigma}_k) /N_k $ with $ N_k := (x_k-x_{k-1}) N $ have a finite logarithmic-moment generating function given by
\[ 
	\Lambda_N(t) := \frac1{N_k} \ln \left(  \frac{1}{2^{N_k}} \sum_{\pmb{\sigma}_k} e^{th_k(\pmb{\sigma}_k) } \right)    =  N_k^{-1} \mkern-15mu\sum_{\lceil x_{k-1} N \rceil +1 \leq j \leq \lceil x_{k} N \rceil}  \mkern-10mu \ln \cosh(th_j)  
	  =:  \ee[\ln\cosh(t \mathfrak{h})  ]  + S_N(t) 
 \]
	where $ S_N(t) $ is a random variable. For any fixed $ t \in \mathbb{R} $ by the strong law of large numbers the latter converges to zero as $ N\to \infty $. In fact, 
	we can find a set of full probability (with respect to the distribution of the iid variables $ (h_i) $) such that the almost-sure convergence
	\[
	\lim_{N\to \infty } \Lambda_N(t) = \ee[\ln\cosh(t \mathfrak{h})  ] 
	\]
	 holds true for all $t\in \mathbb{R} $ simultaneously. This follows from an $ 3 \varepsilon $-argument by considering a countable dense set first and extending this assertion by noticing that both sides are continuously differentiable in $ t $ (see the proof of Lemma~5 in \cite{AK14}).
	 The G\"artner Ellis theorem (cf.~\cite{DemZei10}) then ensures that 
\begin{equation}\label{eq:rate}
I(z) \coloneqq \sup_{t \in \rr} \{zt - \ee[\ln \cosh t \mathfrak{h}]   \}  
\end{equation}
is a rate function for $ N_k^{-1} h_k(\pmb{\sigma}_k) $ for any $ k $. As a Legendre transform $ I : \mathbb{R} \to \mathbb{R} \cup \{ \infty\} $ is lower semicontinuous. It is straightforward to check that 
$ I $ is symmetric, $ I(-z) = I(z) $, equal to $+\infty$ for  $ |z|> \ee[|\mathfrak{h}|] $, continuously differentiable on $  (-\ee[|\mathfrak{h}|], \ee[|\mathfrak{h}|]) $, where it is bounded by $ \ln 2 $, and continuous and monotone on $  [0, \ee[|\mathfrak{h}|]) $.

The G\"artner Ellis theorem also allows to determine the asymptotic behaviour of occupation numbers $ N(y_k) $, which we can rewrite as $ 2^{N_k} $ times the probability that 
\[ h_k/N_k\leq -  y_k/(x_k-x_{k-1}) =: \xi_k(y_k) =:\xi_k . \]
More precisely, we almost surely have
\begin{align}\label{eq:entf1}
 \ln 2 - \inf_{z< -\xi_k} I(z) & \leq \liminf_{N \to \infty} \frac1{N_k} \ln  N(y_k)  \\
 & \leq \limsup_{N \to \infty} \frac1{N_k} \ln  N(y_k) \leq  \ln 2 - \inf_{z\leq -\xi_k} I(z)   = \ln 2 - I(  \xi_k)  . \notag 
 \end{align}
 By the aforementioned continuity of $ I $, we thus obtain for $ \xi_k \in (-\ee[|\mathfrak{h}|], \ee[|\mathfrak{h}|]) $  the almost-sure convergence
 \begin{equation}\label{eq:entf2}
  \lim_{N \to \infty} \frac1{N_k} \ln  N(y_k)  = \ln2 - I(\xi_k) ,
 \end{equation}
 which describes the energy statistics of the magnetic field.
As a next step, we analyse the energy statistics of the total Hamiltonian.  We start by extending our definition of occupation numbers and introduce:
\begin{equation}\label{eq:num} \begin{split}
 N(\pmb{E},\pmb{y}) &:= N(E_1,\ldots,E_n,y_1,\ldots,y_n) \\& \coloneqq |\{ \pmb{\sigma} \in \mathcal{Q}_N | \quad \sqrt{a_k} X_{\pmb{\sigma}_1 \cdots \pmb{\sigma}_k} \leq - \sqrt{N} E_k \text{ and } h_k(\pmb{\sigma}_k) \leq - N y_k \text{ for all } k=1,\ldots,n   \}|  \end{split}  \end{equation} 
Our next goal is to obtain the asmptotics for $ N(\pmb{E},\pmb{y})$. To this end, we introduce the entropy 
 \begin{equation}\label{eq:entropy}
 S(\pmb{E},\pmb{y}) \coloneqq \ln 2 - \sum_{j=1}^{n} \left(\frac{E_j^2}{2a_j}  +(x_j - x_{j-1})I(\xi_j(y_j)) \right)
 \end{equation}
as well as the constraints 
 \begin{equation}\label{eq:cont}
 \mathcal{C} \coloneqq \Big\{(\pmb{E},\pmb{y}) \in \rr_{\geq 0}^n \times \rr_{\geq 0}^n\ \Big| \ \sum_{j=1}^{k} \frac{E_j^2}{2a_j} +(x_j - x_{j-1})I(\xi_j(y_j)) < x_k \ln 2 \; \text{for all } k=1,\ldots,n    \Big\} . 
 \end{equation}
Note that $(\pmb{E},\pmb{y}) \in \mathcal{C}$ guarantees that $ I(\xi_k) < \infty $ for all $ k $.
By continuity of the involved functions on the domain, where they are finite, we conclude that $  \mathcal{C}  $ is an open set and $ \xi_j(y_j) \in (-\ee[|\mathfrak{h}|], \ee[|\mathfrak{h}|]) $ for any $ j $ in case $(\pmb{E},\pmb{y}) \in \mathcal{C}$. 

The following lemma on the asymptotics of $ N(\pmb{E},\pmb{y})$  is a natural generalization of Theorem 1.2 in~\cite{AP19}. We remark that $\frac1N \ln  N(\pmb{E},\pmb{y})$ is shown to converge almost surely, but not in expectation. As the event $\{N(\pmb{E},\pmb{y}) = 0\}$ has a small, but nonvanishing, probability, we in fact have  $\ee[\ln  N(\pmb{E},\pmb{y})] = - \infty$. 

 \begin{lemma}\label{lem:entropy}
 Let $X$ be an $n$-level GREM vector as in \eqref{eq:ugrem} and $(h_i)$ independent copies of an absolutely integrable random variable $\mathfrak{h}$. Then, if  $(\pmb{E},\pmb{y}) \in \mathcal{C}$, we almost surely have 
 \begin{equation}\label{eq:entc1}
 \lim_{N \to \infty} \frac1N \ln  N(\pmb{E},\pmb{y}) =  S(\pmb{E},\pmb{y}).
 \end{equation}
 On the other hand, if  $(\pmb{E},\pmb{y}) \notin \bar{\mathcal{C}}$, the topological closure of $ \mathcal{C}$, almost surely and for all, but finitely many $ N $:
 \begin{equation}\label{eq:entc2}
 N(\pmb{E},\pmb{y})  = 0 . 
 \end{equation}
 \end{lemma}

\begin{proof}
	Let us start with the case $ (\pmb{E},\pmb{y}) \notin \bar{\mathcal{C}}$. One then finds some $k \in \nn$ and $\epsilon > 0$ such that 
	\begin{equation}\label{ass:notC} \sum_{j=1}^{k} \frac{E_j^2}{2a_j} +(x_j - x_{j-1})I(y_j/(x_j - x_{j-1})) \geq x_k \ln 2 + \epsilon.   \end{equation}
	We condition on the weights $(h_i)$ and compute the probability that a reduced spin vector $\pmb{\sigma}_1 \cdots \pmb{\sigma}_k$ meets the first $k$ energy requirements
	\begin{equation}\label{eq:Gaussian}
	 \begin{split} &\pp(\sqrt{a_j } X_{\pmb{\sigma}_1 \cdots \pmb{\sigma}_j} \leq - \sqrt{N} E_j \text{ and } h_j(\pmb{\sigma}_j) \leq - N y_j \text{ for all } j=1,\ldots,k \, | (h_i) )  \\ & = \prod_{j =1}^{k}  \pp(\sqrt{a_j } X_{\pmb{\sigma}_1 \cdots \pmb{\sigma}_j} \leq - \sqrt{N} E_j) \, \pp( h_j(\pmb{\sigma}_j) \leq - N y_j \, | (h_i) ) \\ &  \leq \prod_{j =1}^{k}  e^{-N E_j^2/ (2a_j)} 1[ h_j(\pmb{\sigma}_j) \leq - N y_j]. \end{split}   
	 \end{equation}
	The first equality is due to the independence of the variables $ X_{\pmb{\sigma}_1 \cdots \pmb{\sigma}_j}  $ for different $ j $. The bound on the first probability follows from the standard Gaussian estimate.  This may be inserted into the following union bound
	\[ \begin{split}  \pp(N(\pmb{E},\pmb{y}) \geq 1| (h_i))  & \leq \sum_{\pmb{\sigma}_1\cdots \pmb{\sigma}_k} \pp\big(\sqrt{a_j N} X_{\pmb{\sigma}_1 \cdots \pmb{\sigma}_j} \leq - N E_j \text{ and } h_j(\pmb{\sigma}_j\big) \leq - N y_j \text{ for all } j=1,\ldots,k \, | (h_i) )  \\
	& \leq \exp\Big( -N \sum_{j=1}^{k} \frac{E_j^2}{2a_j} \Big) \, \prod_{j=1}^k N(y_k) ,
	\end{split} 
	\]
	where the last inequality is the previous estimate.
	
	We now distinguish two cases. If  
	\[  \pmb{y} \in  \mathcal{G}_k  \coloneqq \{\pmb{y} \in \rr_{\geq 0}^n| \, I( \xi_j(y_j) ) < \infty \; \mbox{for all $j=1,\dots , k $}\}, \] 
	we may further estimate the right side using~\eqref{ass:notC}  and the upper bound in~\eqref{eq:entf1} to conclude that for all, but finitely many $N $ and almost surely with respect to the variables $(h_i) $:
	\[
	  \pp(N(\pmb{E},\pmb{y}) \geq 1| (h_i)) \leq e^{-N\epsilon/2}  .
	\]
	A Borel-Cantelli argument then shows that $N(\pmb{E},\pmb{y})$ converges almost surely to zero.
	
	In case $  \pmb{y} \not\in  \mathcal{G}_k $ there exist an integer  $j \in \{1, \dots  k\}  $ such that $I(  \xi_j(y_j)) = \infty $.   Consequently, \eqref{eq:entf1} implies the almost-sure convergence $  \limsup_{N \to \infty} \frac1{N_j} \ln  N(y_j)  = -\infty $. Since $ N(y_j) \in \mathbb{N}_0 $, this implies that almost surely $ N(y_j) = 0 $ for all, but finitely many $ N $. In turn. we conclude $ \pp(N(\pmb{E},\pmb{y}) \geq 1| (h_i)) = 0 $ for all, but finitely many $ N $ and hence the claim~\eqref{eq:entc2} in this case.\\

It thus remains to show \eqref{eq:entc1} for $ (\pmb{E},\pmb{y}) \in \mathcal{C}$ . This proof will be based on
Proposition~\ref{claim} below. For its application, we introduce the following sequences of numbers 
\[ \begin{split}
F_k(N) & \coloneqq \frac1N \ln |\{ \pmb{\sigma}_1 \cdots \pmb{\sigma}_k  | \; \sqrt{a_i } X_{\pmb{\sigma}_1 \cdots \pmb{\sigma}_i} \leq - \sqrt{N} E_i  \text{ and } \\
& \mkern200mu   h_j(\pmb{\sigma}_j) \leq - N y_j \text{ for all } i=1,\ldots,k-1; \; j=1,\ldots,k   \}| \\
G_k(N) & \coloneqq \frac1N \ln |\{ \pmb{\sigma}_1 \cdots \pmb{\sigma}_k  | \; \sqrt{a_i } X_{\pmb{\sigma}_1 \cdots \pmb{\sigma}_i} \leq - \sqrt{N} E_i \text{ and } \\
& \mkern200mu  h_j(\pmb{\sigma}_j) \leq - N y_j \text{ for all } i=1,\ldots,k; \; j=1,\ldots,k   \}| , \quad G_0 := 0 .
\end{split}     \]
The definition of these sets are motivated by inclusion-exclusion. If we suppose that $G_k(N) $ converges as $ N \to \infty $, the almost-sure convergence~\eqref{eq:entf2}, for which we recall that $ (\pmb{E},\pmb{y}) \in \mathcal{C}$ implies $\max_{j} |\xi_j| <  \ee[|\mathfrak{h}|] $, yields
\[ \lim_{N \to \infty} F_{k+1}(N) = (x_{k+1}-x_{k}) \ln 2 - (x_{k+1} - x_{k})I(\xi_{k+1}) + \lim_{N \to \infty} G_{k}(N)  . \]
Moreover, Proposition~\ref{claim} below further implies 
\[ \lim_{N \to \infty} G_{k+1}(N) = -(2a_{k+1})^{-1} E_{k+1}^2 + \lim_{N \to \infty} F_{k+1}(N),   \]
provided that the right side is positive. By definition of the constraint, this is always the case if $(\pmb{E},\pmb{y}) \in \mathcal{C}$ such that
\[ \lim_{N \to \infty}  \frac1N \ln  N(\pmb{E},\pmb{y}) = \lim_{N \to \infty} G_n(N) =  \ln 2 - \sum_{j=1}^{n} \left(\frac{E_j^2}{ 2a_j} +(x_j - x_{j-1})I\left(\xi_j(y_j)\right) \right) =  S(\pmb{E},\pmb{y}) \]
almost surely.
\end{proof}

	The second part of the proof of Lemma~\ref{lem:entropy} relied on the following claim, whose proof follows that of Proposition 6 in \cite{AK14}. 
	\begin{proposition}\label{claim}
		Let $(D_N)_{N \in \nn}$ be a family of finite sets, $(X_s)_{s \in D_N}$ independent standard Gaussian variables and $(Y_s)_{s\in D_N} $ a random vector, which is independent of $X$ and whose entries only take the values $0$ and $1$. Further, suppose that almost surely
		\[ \lim_{N \to \infty} \frac1N \ln |\{s \in D_N | Y_s  = 0 \}| =  q > 0.  \] 
		Then the number of large deviations 
		\[ N(E) \coloneqq |\{s \in D_N | Y_s  = 0 \text{ and } \sqrt{a} X_s \leq -E \sqrt{N} \}|,   \]
		with $a > 0$ almost sure obeys
		\[ \lim_{N \to \infty}  \frac{1}{N} \ln N(E) = q - (2a)^{-1} E^2 \] 
		provided that $q >(2a)^{-1} E^2$. 
	\end{proposition}
\begin{proof}
We apply the second moment method to $ N(E) $ and the conditional expectation conditioned on the event
$ Z \coloneqq \{s \in D_N | Y_s  = 0 \} $.
A standard calculation similar to~\eqref{eq:Gaussian} using elementary bounds on the Gaussian distribution function shows that 
\[  \ee[N(E)|Z] = Z \ \exp(-((2a)^{-1} E^2 + o(1))N). \]
By explicit computation we determine the second moment of $ N(E) $ conditioned on $ Z $:
\[ \begin{split}
& \ee[N(E)^2|Z] - \ee[N(E)|Z]^2\\
 &= \mkern-10mu \sum_{s,s^\prime: Y_s = Y_{s^{\prime}} =0} \mkern-15mu \pp\left( \sqrt{a} X_s \leq -E\sqrt{N} \; \mbox{and} \, \sqrt{a} X_{s^\prime} \leq -E\sqrt{N} \right) -\pp\left(\sqrt{a} X_s \leq -E\sqrt{N}\right) \pp\left( \sqrt{a} X_{s^\prime} \leq -E\sqrt{N}\right)   \\
&= \sum_{s: Y_s=0} \pp\left(\sqrt{a} X_s \leq -E\sqrt{N}\right) - \pp\left(\sqrt{a} X_s \leq -E\sqrt{N} \right) ^2 \leq \ee[N(E)|Z].
\end{split}  \]
Thus, the Chebyshev inequality implies for any $\epsilon > 0$:
\[ \pp(| N(E)-  \ee[N(E)|Z]| > \epsilon \ \ee[N(E)|Z] | Z) \leq \epsilon^{-2} \ \ee[N(E)|Z]^{-1} .   \]
We note that $\ee[N(E)|Z]$ is almost surely exponentially large; in fact,  by assumption $\ln Z = N( q +o(1) ) $ with  $q >(2a)^{-1} E^2$. Thus, a Borel-Cantelli argument yields almost surely
\[ \limsup_{N \to \infty} \left| \frac1N \ln \frac{N(E)}{\ee[N(E)|Z]} \right| = 0,  \]
which completes proof using the expression for $\ee[N(E)|Z]$.
\end{proof}

Based on Lemma~\ref{lem:entropy}, we may now establish a variational expression for the limiting pressure of the $n $-level GREM in a random magnetic field.
\begin{lemma}\label{lem:var}
	For any $\beta \geq 0$ and any absolutely integrable random variable $\mathfrak{h}$ the pressure $\Phi_N(\beta,\mathfrak{h})$ converges almost surely and its limit is given by 
	\begin{equation}\label{eq:var}
	\lim_{N \to \infty} \Phi_N(\beta,\mathfrak{h}) = \sup_{(\pmb{E},\pmb{y}) \in \mathcal{C}} \left( \beta(E_1+\cdots+E_n+y_1+\cdots+y_n) + S(\pmb{E},\pmb{y}) \right).
	\end{equation}
\end{lemma} 

\begin{proof}
	By elementary estimates it follows that \[\exp(N \Phi_N(\beta,\mathfrak{h})) \geq 
	\exp(\beta N(E_1+\cdots+E_n+y_1+\cdots+y_n )) N(\pmb{E},\pmb{y})\]
	for any $(\pmb{E},\pmb{y})$, 
	which in view of Lemma~\ref{lem:entropy} implies almost surely
	\[\liminf_{N \to \infty} \Phi_N(\beta,\mathfrak{h}) \geq \sup_{(\pmb{E},\pmb{y}) \in \mathcal{C}} \beta(E_1+\cdots+E_n+y_1+\cdots+y_n) + S(\pmb{E},\pmb{y}).  \]
	To obtain an asymptotic upper bound we use a discretization argument. We set $\alpha \coloneqq \max_{i=1,\ldots,n} a_i$ and define the compact box 
	\[ F  \coloneqq [-(\sqrt{2\alpha \ln2}+1),\sqrt{2\alpha \ln2}+1]^n \times [-\ee[|\mathfrak{h}|]-1,\ee[|\mathfrak{h}|]+1]^n.   \]
	One easily sees that almost surely no configuration $(\pmb{E},\pmb{y})$ outside of $F$ contributes to the limit~\eqref{eq:var} of the pressure. To simplify the notation, we assume  in the following that this holds true for any $N$. Thus, it suffices to consider configurations in $F$ on which we set the grid 
	\[ F_K \coloneqq \left\{(\pmb{E},\pmb{y}) \in F \ \big| \ E_j = \frac{k_j}{K} (\sqrt{2\alpha \ln2}+1), \, y_j = \frac{l_j}{K} (\ee[\mathfrak{h}]+1), \quad \begin{array}{c} k_j,l_j = -K,-K+1,\ldots,K, \\  j=1,\ldots,n  \end{array}  \right\}    \]
	with $K \in \nn$. We pick  $\epsilon > 0$ arbitrary and choose $K=K_\epsilon$ such that $\max\{\ee[\mathfrak{h}]+1,\sqrt{2\alpha \ln2}+1\} < \epsilon K_\epsilon $. Then, the $\epsilon$-neighborhoods of the grid points in $F_K$ cover the box $F$ and therefore
	\[ e^{N \Phi_N(\beta,\mathfrak{h})} \leq \sum_{(\pmb{E},\pmb{y}) \in F_K} \mkern-5mu N(\pmb{E},\pmb{y}) \ e^{\beta N(E_1+\cdots+E_n+y_1+\cdots+y_n + 2n \epsilon)}. \]
	Let us now observe three points. First, if $E_j$ or $y_j$ is negative for some $j$ we may replace this value by $0$ without changing the number $N(\pmb{E},\pmb{y})$ on an exponential scale. This is just a consequence of symmetry and the LDP satisfied by $h_j$ and the Gaussian vectors $ X $. Secondly, without loss of generality we may assume that there are no grid points on the boundary of $\mathcal{C}$. Moreover, if 
	$(\pmb{E},\pmb{y}) \notin \bar{\mathcal{C}}$, the corresponding term gives no contribution to the limit of $\Phi_N$ by \eqref{eq:entc2}. Thirdly, the entropy factor corresponding to the summation over the grid points does not depend on $N$ and is thus irrelevant after taking the limit. Summarizing these points, we conclude almost surely
	\[ \limsup_{N \to \infty} \Phi_N(\beta,\mathfrak{h}) \leq 2\beta n \epsilon + \sup_{(\pmb{E},\pmb{y}) \in \mathcal{C}} \beta(E_1+\cdots+E_n+y_1+\cdots+y_n) + S(\pmb{E},\pmb{y}),  \]
	which completes the proof as $\epsilon > 0$ was chosen arbitrarily.
\end{proof}

It remains to solve the variational problem \eqref{eq:var} which is the last part in the proof of Theorem~\ref{prop:gremh}. Note that one may replace the $\sup$ on $\mathcal{C}$ by a maximum on $\bar{\mathcal{C}}$ as the involved expressions possess continuous extensions to $\bar{\mathcal{C}}$. 
\begin{proof}[Proof of Theorem~\ref{prop:gremh}]
We proceed via induction on $m$, the number of linear segments of the concave hull $\bar{A}$. If $m = 1$, the variational problem consists of $2n$ independent optimisation problems which can be easily solved. This leads to 
\[ E_j = \beta a_j \text{ and } y_j = (x_j-x_{j-1}) \ee[ \mathfrak{h} \tanh(\beta \mathfrak{h})] \quad j=1,\ldots,n. \]
To obtain the expression for $y_j$, it is helpful to note that the rate function $I$ is the Legendre transform of $\ee[\ln \cosh(\beta \mathfrak{h})]$. The maximum is attained when $\xi_j(y_j) = y_j/(x_j-x_j-1) $ equals the derivative of  $\ee[\ln \cosh(\beta \mathfrak{h})]$  with respect to $\beta$. We see that if $\beta$ is small enough, all constraints are fulfilled and the maximum is given by 
\[ \Phi(\beta,\mathfrak{h}) = \ln 2 + \frac{\beta^2}{2} \left( \sum_{j =1}^{n} a_j \right) + \ee[\ln \cosh(\beta \mathfrak{h})].   \]
Since in the unconstrained variational problem the optimal value $E_j$ is unbounded as $\beta$ increases, the above considerations will hold true up to some critical value $\beta_c$, where the first constraining inequality is not satisfied, i.e., the maximum is located at the boundary of $\mathcal{C}$. Due to the structure of the optimal $(\pmb{E},\pmb{y})$ in the unconstrained setting, this needs to be the inequality corresponding to the highest slope $(a_1+\cdots+a_k)/x_k$ which is attained at $k = n$ since $m = 1$. If we denote the optimal configuration of the unconstrained problem at $\beta_c$ by $(\pmb{E}^c,\pmb{y}^c)$ we thus have 
\[  S(\pmb{E}^c,\pmb{y}^c) = 0.  \]	
From there, one obtains after some algebra  the self-consistency equation for $\beta_c$:
\[ \frac{\sum_{j} a_j}{2} \beta_c^2 =  \ln 2 + \ee[\ln \cosh \beta_c \mathfrak{h}] - \beta_c\ee[\mathfrak{h} \tanh \beta_c \mathfrak{h}].    \]
Furthermore,
\[ \max_{(\pmb{E},\pmb{y}) \in \bar{\mathcal{C}}} \beta_c(E_1+\cdots+E_n+y_1+\cdots+y_n) + S(\pmb{E},\pmb{y}) = \max_{(\pmb{E},\pmb{y}) \in \bar{\mathcal{C}}} \beta_c(E_1+\cdots+E_n+y_1+\cdots+y_n),  \]
which is clearly still a valid identity for $\beta > \beta_c$. 
We conclude that $\Phi$ becomes a linear function of $\beta$ for $\beta > \beta_c$ and the slope agrees with the derivative of $\Phi$ at $\beta_c$, i.e.
\[ \Phi(\beta,\mathfrak{h}) = \beta\left( \beta_c \sum_{j=1}^n a_j + \ee[\mathfrak{h} \tanh \beta_c \mathfrak{h}]) \right).  \]
This is exactly the statement of Theorem~\ref{prop:gremh} in the case $m =1$.

Now, suppose that the assertions are true for some $m$, and we want to show that it is still the case for $m+1$.
Let us write $\pmb{E}_{<m},\pmb{E}_{>m}, \pmb{y}_{<m}$ and $\pmb{y}_{>m}$, where the vectors denote the energy configurations corresponding to the first $m$ segments and the last segment, respectively.  Similarly, we set $\mathcal{C}_m$ the set of the constraints related to the first $m$ segments.
If we only demand that the energy configuration $(\pmb{E}_{<m},\pmb{y}_{<m})$ satisfy the constraints $\mathcal{C}_m$, then using the induction hypothesis and the analysis of the case $m=1$, we end up with the expression 
\[ \sum_{l=1}^{m} \varphi^{(l)}(\beta,\mathfrak{h}) + (1-y_m)\ln 2 + \frac{\beta^2}{2} \left( \sum_{j \in I_{m+1}} a_j \right) + (1-y_m) \ee[\ln \cosh(\beta \mathfrak{h})] \] 
for the limit of the pressure, where $I_{m+1}$ denotes the last segment. This is indeed a solution if $\beta \leq \beta^{(m)}_c$, since the remaining constraints are also verified by the $m$-level solution $(\pmb{E}_{<m},\pmb{y}_{<m})$ and the unconstrained solution $(\pmb{E}_{>m},\pmb{y}_{>m})$ due to the concave-hull structure. We note that for $\beta > \beta^{(m)}_c$, we only need to consider the $n$-th inequality (for the same reason as in the case $m =1$) which then may be rewritten as 
\[  (y_{m+1}-y_m)  \ln 2 > \sum_{j \in I_{m+1}} (2a_j)^{-1} +(x_j - x_{j-1})I(\xi_j(y_j)) .  \]
Thus, the situation is analogous to the case  $m=1$  and the same arguments lead to the expression for $\beta_c^{(m+1)}$ and the pressure $\Phi$ if $\beta > \beta_c^{(m+1)}$.
\end{proof}   

\subsection{From GREM to QGREM: application of the peeling principle}

We now consider the QGREM with a random magnetic field in $z$-direction as in Theorem~\ref{thm:qcremh}.

\begin{theorem}\label{prop:qgremh}
	Let $U(\pmb{\sigma})$ be a GREM process as in \eqref{eq:ugrem},  $\beta \geq 0$ and $\mathfrak{h},\mathfrak{b}$  absolutely integrable random variables. Then, almost surely
	\begin{equation}\label{eq:qgremh}
	\lim_{N \to \infty} \Phi_N(\beta,\mathfrak{h},\mathfrak{b}) = \max_{0 \leq k \leq m} \left( \sum_{l=1}^{k} \varphi^{(l)}(\beta,\mathfrak{h}) + (1-y_k) \ee[\ln 2 \cosh(\sqrt{\mathfrak{b}^2+\mathfrak{h}^2})] \right) .
	\end{equation}
	Here, the empty sum is interpreted to be zero.
\end{theorem}

\begin{proof}
	We recall the definition of the cut GREM $ U^{(x_k)} := U^{(0,x_k)}$ which may be represented as 
	\[ U^{(x_k)}(\pmb{\sigma}_1\pmb{\sigma}_2\cdots \pmb{\sigma}_k)  = \sqrt{a_1} X_{\pmb{\sigma}_1} + \sqrt{a_2}X_{\pmb{\sigma}_1\pmb{\sigma}_2} + \cdots + \sqrt{a_n}X_{\pmb{\sigma}_1\pmb{\sigma}_2\cdots \pmb{\sigma}_k}. \] 
	An iterative application of the peeling principle (\cite[Theorem 2.3]{MW20c}; see also the proof of \cite[Corollary 2.7]{MW20c}) yields almost surely
	\[ \limsup_{N \to \infty} \left|\Phi_N(\beta,\mathfrak{h},\mathfrak{b}) - \max_{0 \leq k \leq n} \frac1N \ln \tr e^{-\beta(U^{(x_k)} - h(\pmb{\sigma}) - B^{(2,x_k)})}\right| = 0.\]
	The cut-magnetic field $B^{(2,x_k)}$ was defined in \eqref{eq:cutB}. We naturally split the longitudinal field,
	\[ h(\pmb{\sigma}) = h^{(1,x_k)}(\pmb{\sigma}_1 \cdots \pmb{\sigma}_k ) + h^{(2,x_k)}(\pmb{\sigma}_{k+1} \cdots \pmb{\sigma}_n ); \quad h^{(1,x_k)}(\pmb{\sigma}_1 \cdots \pmb{\sigma}_k ) \coloneqq \sum_{i =1}^{\lceil x_k N \rceil} h_i \sigma_i \] 
	 and apply Theorem~\ref{prop:gremh} to the Hamiltonian 
	 $H^{(x_k)} \coloneqq U^{(x_k)} - h^{(1,x_k)}$. Together with the strong law of large numbers for $ h^{(2,x_k)}(\pmb{\sigma}_{k+1} \cdots \pmb{\sigma}_n )+ B^{(2,x_k)}$. Thus we arrive at 
 \begin{equation}\label{eq:limx}
  \lim_{N \to \infty} \Phi_N(\beta,\mathfrak{h},\mathfrak{b})
  = \max_{0 \leq k \leq n} \left( \Phi^{(x_k)}(\beta,\mathfrak{h}) + (1-x_k)  \ee[\ln 2 \cosh(\sqrt{\mathfrak{b}^2+\mathfrak{h}^2})] \right), \end{equation}
  where $\Phi^{(x_k)}(\beta,\mathfrak{h})$ denotes the limit of the pressure of the Hamiltonian  $H^{(x_k)}$ restricted 
  to the Hilbert space of subgraph $\mathcal{Q}_{\lceil x_k N \rceil}$ spanned by $\pmb{\sigma}_1 \cdots \pmb{\sigma}_k$. (Note that for $H^{(x_k)}$ on the total graph $\mathcal{Q}_N$ the resulting pressure is $\Phi^{(x_k)}(\beta,\mathfrak{h}) + (1-x_k) \ln 2)$.)
  
 If the cut point coincides with and endpoint of the concave hull. i.e. $x_k = y_j$ for some $j$, we have 
  \[ \Phi^{(y_j)}(\beta,\mathfrak{h}) =  \sum_{l=1}^{j} \varphi^{(l)}(\beta,\mathfrak{h}) . \]
 Thus, it only remains to show that the maximum in \eqref{eq:limx} is attained at some $y_l$. We follow the comparison argument presented in the proof of \cite[Lemma 3.1]{MW20c}.
  If $\{ x_0,\ldots, x_n \} = \{ y_0,\ldots, y_m \}$, the assertion is trivial. So, let $y_l < x_k < y_{l+1}$. We recall that distribution function $A^{(x_k)}$ associated with $U^{(x_k)}$ is given by
  \[ A^{(x_k)} = \begin{cases}
  A(x) & \text{ if } x \leq x_k, \\
  A(x_k) & \text{ else.}
  \end{cases}.  \]
  We introduce the Gaussian processes $Y$ and $Z$ of GREM type with the distribution functions 
  \[ A_Y(x) \coloneqq \begin{cases}
  A(x) & \text{ if } x \leq y_l, \\
  A(y_l) & \text{ if } y_l < x < x_k ,\\
  A(x_k) & \text{ if } x \geq x_k ,
  \end{cases} 
\quad A_Z(x) \coloneqq \begin{cases}
  A(x) & \text{ if } x \leq y_l, \\
  A(y_l) & \text{ if } y_l < x < x_k, \\
  A(y_l) + \frac{x_k - y_l}{y_{l+1} - y_l} (A(y_{l+1}) - A(y_l) , & \text{ if } x \geq x_k.
  \end{cases}  \] 
  which shall be independent of the weights  $(h_i)$
After conditioning on the random weights $(h_i)$,   Slepian's lemma (cf.~\cite{Bov06}) and the independence of $(h_i)$ and the GREM processes imply: 
  \begin{align}\label{eq:Slepian} \lim_{N \to \infty} \frac1N \ln \tr_{| \ell^2(\mathcal{Q}_{x_k N})} e^{-\beta(U^{(x_k)}- h^{(1,x_k)})} &\leq \lim_{N \to \infty} \frac1N \ln \tr_{| \ell^2(\mathcal{Q}_{x_k N})} e^{-\beta (\sqrt{N} Y - h^{(1,x_k)})} \notag \\ 
  &\leq \lim_{N \to \infty} \frac1N \ln \tr_{| \ell^2(\mathcal{Q}_{x_k N})} e^{-\beta (\sqrt{N} Z - h^{(1,x_k)})}. 
   \end{align}
For the second inequality, we recall that $A$ is majorized by its concave hull $\bar{A}$ and agrees with $\bar{A}$ at $y_l$ and $y_{l+1}$:
  \[ A(x_k) \leq A(y_l) + \frac{x_k - y_l}{y_{l+1} - y_l} \left(A(y_{l+1}) - A(y_l) \right).  \]
  Since the pressure is an increasing function of the jump heights (cf.~\eqref{eq:part}), we hence arrive at the second bound in~\eqref{eq:Slepian}. The resulting pressure is computed easily in terms of the partial pressures \eqref{eq:part} corresponding to $ A $:
  \[ \lim_{N \to \infty} \frac1N \ln \tr_{| \mathcal{Q}_{x_k N}} e^{-\beta (\sqrt{N} Z - h^{(1,x_k)}) } = \sum_{j = 1}^{l} \varphi^{(j)}(\beta,\mathfrak{h})+ \frac{x_k-y_l}{y_{l+1} - y_l} \varphi^{(j+1)}(\beta,,\mathfrak{h}). \]
Using the abbreviation $p(\beta,\mathfrak{h},\mathfrak{b}) \coloneqq \ee[\ln 2 \cosh(\beta \sqrt{ \mathfrak{b}^2+\mathfrak{h}^2})]$ we thus conclude
  \[\begin{split} & \lim_{N \to \infty} \frac1N \ln\tr_{| \ell^2(\mathcal{Q}_{x_k N})}  e^{-\beta (U^{(x_k)}- h^{(1,x_k)})} + (1-x_k) p(\beta,\mathfrak{h},\mathfrak{b}) \\ &\leq \sum_{j = 1}^{l} \varphi^{(j)}(\beta,\mathfrak{h}) + (1-y_{l}) p(\beta,\mathfrak{h},\mathfrak{b}) + \frac{x_k-y_l}{y_{l+1} - y_l} \left( \varphi^{(l+1)}(\beta,\mathfrak{h}) - (y_{l+1} - y_l)  p(\beta,\mathfrak{h},\mathfrak{b}) \right). \end{split} \]
  Depending on the sign of the term in the last bracket, we have 
  \[  \lim_{N \to \infty} \frac1N \ln\tr_{| \ell^2(\mathcal{Q}_{x_k N})}  e^{-\beta (\sqrt{N} X^{(x_k)}- V^{(1,x_k)})} + (1-x_k) p(\beta,\mathfrak{h},\mathfrak{b}) \leq \sum_{j = 1}^{l} \varphi^{(j)}(\beta,\mathfrak{h}) + (1-y_{l}) p(\beta,\mathfrak{h},\mathfrak{b})  \]
or the sum on the right side runs to $ l+1 $ and $ y_l $ is replaced by $ y_{l+1} $.  
  Consequently, the maximal pressure is indeed attained at some $y_l$.	
\end{proof}

\subsection{Finishing the proof: the interpolation argument}

Finally, we will lift Theorem~\ref{prop:qgremh} to the case of a general QGREM. The idea is to show that the left and right side of \eqref{eq:qcremh} are continuous with respect to the distribution function $A$ and the uniform norm. We start with the continuity of the right side, i.e., spelling out explicitly the $ A $-dependence of quantities for the moment, we need to show that 
\[
 \Phi(\beta,\mathfrak{h},\mathfrak{b},A) = \sup_{0 \leq z \leq 1} \left( \int_{0}^{z} \varphi(\beta,\mathfrak{h},A,x) \, dx + (1-z) \ee[\ln 2 \cosh(\beta \sqrt{\mathfrak{b}^2+\mathfrak{h}^2})] \right),
\]
is continuous in $A$. We recall that the density is given by
\[
\varphi(\beta,\mathfrak{h},A,x) \coloneqq \begin{cases} 
\ln 2 + \bar{a}(x)\frac{\beta^2}{2} + \ee[\ln \cosh \beta \mathfrak{h}] & if \quad \beta \leq \beta_c(A,x) \\
\beta(\bar{a}(x)\beta_c(A,x) + \ee[\mathfrak{h} \tanh \beta_c(A,x) \mathfrak{h}]) & if \quad \beta > \beta_c(A,x) 
\end{cases}
\]
where the critical temperature $\beta_c(A,x))$ is the unique positive solution of the self-consistency equation 
\[
\frac{\bar{a}(x)}{2}\beta_c(A,x)^2 =  \ln 2 + \ee[\ln \cosh \beta_c(A,x) \mathfrak{h}] - \beta_c(A,x) \ee[\mathfrak{h} \tanh \beta_c(A,x) \mathfrak{h}]. 
\]

\begin{lemma}\label{lem:cont}
	Let $\beta \geq 0$ and $\mathfrak{b},\mathfrak{h}$ be absolutely integrable random variables. Moreover, let $(A_n)_{n \in \nn}$, $A$ be distribution functions on $[0,1]$ such that $A_n$ converges uniformly to $A$. Then,
	\begin{equation}
	\lim_{n \to \infty} \Phi(\beta,\mathfrak{h},\mathfrak{b},A_n) = \Phi(\beta,\mathfrak{h},\mathfrak{b},A).
	\end{equation}
\end{lemma}

\begin{proof}
	It suffices to show that 
	\[ \lim_{n \to \infty} \int_{0}^{1} |\varphi(\beta,\mathfrak{h},A,x) - \varphi(\beta,\mathfrak{h},A_n,x)| \, dx = 0. \]
	We first prove that the integrand converges almost everywhere (with respect to the Lebesgue measure and $x$)
	to zero. 
	One easily sees that the concave hulls $\bar{A}_n$ converge uniformly to  $\bar{A}$ and the right derivatives $\bar{a}_n(x)$ converge to $\bar{a}(x)$ at any $x$, where $\bar{a}(x)$ is continuous (cf.~the proof of Lemma 3.3 in \cite{MW20c}). Since $\bar{A}$ is concave, this ensures that $\bar{a}_n(x)$ converge almost everywhere to $\bar{a}(x)$. Next, we observe that $\beta_c(x,A)$ is a continuous function of $\bar{a}(x)$ by the implicit function theorem and, thus, $\beta_c(x,A_n)$ converges almost everywhere to $\beta_c(x,A)$. This implies that $\varphi(\beta,\mathfrak{h},A_n,x)$ converges almost everywhere. Now we pick some $\delta >0$ and note that the sequence $\varphi(\beta,\mathfrak{h},A_n,x)$ is uniformly bounded due to the general bound 
	\[ 0 \leq \varphi(\beta,\mathfrak{h},A_n,x) \leq  \ln 2 + \bar{a}_n(x)\frac{\beta^2}{2} + \ee[\ln \cosh \beta \mathfrak{h}] \]
	and the monotonicity of the derivatives $\bar{a}_n(x)$. We conclude that for any $\delta > 0$
	\[  \lim_{n \to \infty} \int_{\delta}^{1} |\varphi(\beta,\mathfrak{h},A,x) - \varphi(\beta,\mathfrak{h},A_n,x)| \, dx = 0 .  \] 
	Using the above bound on $[0,\delta]$, we obtain 
	\[  \int_{0}^{\delta} |\varphi(\beta,\mathfrak{h},A,x) - \varphi(\beta,\mathfrak{h},A_n,x)| \, dx \leq \delta(2 \ln 2 + (\bar{A}_n(\delta)+\bar{A}(\delta))\frac{\beta^2}{2} + 2 \ee[\ln \cosh \beta \mathfrak{h}]) \]
	which vanishes if we take the limit $n \to \infty$ and then $\delta \to 0$.
\end{proof}

We turn to the interpolation argument for the left side in \eqref{eq:qcremh}. Let $U,U^\prime$ be two GREM processes with distribution functions $A$, $A^\prime$ and pressures $\Phi(\beta,\mathfrak{h},\mathfrak{b},A)$, $\Phi(\beta,\mathfrak{h},\mathfrak{b},A^\prime)$.  From \cite[Equation (2.16)]{MW20c} we conclude
\begin{equation}\label{eq:inter} |\ee[\Phi(\beta,\mathfrak{h},\mathfrak{b},A) - \Phi(\beta,\mathfrak{h},\mathfrak{b},A^\prime)]| \leq \beta^2 \|A - A^\prime \|_{\infty}, \end{equation}
The Gaussian concentration inequality (cf.~\cite[Proposition 2.9]{MW20c}) guarantees the almost-sure convergence
\[ \limsup_{N \to \infty} |\ee[\Phi(\beta,\mathfrak{h},\mathfrak{b},A)] - \Phi(\beta,\mathfrak{h},\mathfrak{b},A) | = 0.  \]  
We are ready to finish the proof of Theorem~\ref{thm:qcremh}:
\begin{proof}[Proof of Theorem~\ref{thm:qcremh}]
We fix $\beta \geq 0$ and absolutely integrable random variables $\mathfrak{b},\mathfrak{h}$ and use the shorthand notations $\Phi(A) \coloneqq \Phi(\beta,\mathfrak{h},\mathfrak{b},A)$. Let $U$ be a GREM process with distribution function $A$. We pick some $\epsilon > 0$ and an finite-level GREM  $U^\prime$ with distribution function $A^\prime$ such that $\|A - A^\prime\|_{\infty} \leq \epsilon$ and $ |\Phi(A) = \Phi(A^\prime)| \leq \epsilon $.
This is possible thanks to Lemma~\ref{lem:cont}.
We then obtain 
\[ \begin{split}
\limsup_{N \to \infty} |\Phi_N(A) - \Phi(A)|& \leq \limsup_{N \to \infty}  |\Phi_N(A) - \Phi_N(A^\prime)| + |\Phi_N(A^\prime) - \Phi(A^\prime)| + |\Phi(A) - \Phi(A^\prime)| \\& \leq (\beta^2+1) \epsilon .
\end{split} \] 
The final line follows from our preparatory estimate \eqref{eq:inter} and Theorem~\ref{prop:qgremh}, which coincides with Theorem~\ref{thm:qcremh} for an $ n $-level GREM. Since $\epsilon > 0$ is arbitrary, this proves \eqref{eq:qcremh}.

The remaining assertions now follow easily: $\varphi(\beta,\mathfrak{h},x)$ is clearly an increasing function of $\bar{a}(X)$ which in turn is decreasing in $x$. Thus, $\varphi(\beta,\mathfrak{h},x)$ is a decreasing function of $x$. Similarly, the critical inverse temperature $\beta_c(x)$ is increasing as it is a decreasing function of $\bar{a}(x)$. Finally, the fact that $\varphi(\beta,\mathfrak{h},x)$ is increasing and convex in $\beta$ directly follows from \eqref{eq:density}.
\end{proof}

\appendix
\section{Proof of Corollary~\ref{cor:qremh} and Proposition~\ref{prop:critbg}}\label{sec:pfqrem}

We start with the straightforward proof of Corollary~\ref{cor:qremh}:
\begin{proof}[Proof of  Corollary~\ref{cor:qremh} ]
	To apply Theorem~\ref{thm:qcremh}, we note that in the case of the QREM we have $\varphi(\beta,h,x) = \Phi^{\mathrm{REM}}(\beta,h)$ for any $x$. So, we directly obtain \eqref{eq:qremh}. It remains to show that the self-consistency equation
	 \[\frac{1}{2}\beta_c^2 =\ln 2 + \ln \cosh \beta_c h - \beta_c h \tanh \beta_c h, \]
	 which get from Theorem~\ref{thm:qcremh} is equivalent to 
	 \eqref{eq:bcrit2}, i.e. $ \beta_c(h)^2 = 2 r(\tanh(\beta_c(h)h)) $. 
	 This follows from the elementary computation 
	 \[ \begin{split} & r(\tanh(x)) = \ln2-\frac12 \left( (1-\tanh(x)) \ln (1-\tanh(x))+ (1+\tanh(x)) \ln (1+\tanh(x))\right  ) \\
	 &= \ln 2 + \ln \cosh x - \frac12\left( (1-\tanh(x)) \ln (\cosh x - \sinh x)+ (1+\tanh(x)) \ln (\cosh x + \sinh x)\right  ) \\
	 &= \ln 2 + \ln \cosh x - x \tanh x \end{split}  \]
	 for any $x \in \rr$.
\end{proof}

The proof of Proposition \ref{prop:critbg} is based on multiple elementary, but quite lengthy, computations. 
\begin{proof}[Proof of Proposition \ref{prop:critbg} ] 
1.~The defining equation \eqref{eq:bcrit2} immediately implies that $\beta_c(h)$ is a strictly decreasing function. The Taylor expansions $r(y) = \ln2 - y^2/2 + \mathcal{O}(y^4))$ and $\tanh(y) = y + \mathcal{O}(y^2)$ yield for small $h > 0$
		\[\frac12 \beta_c(h)^2 = \ln2 - \frac{(\beta_c(h)h)^2}{2} + \mathcal{O}(h^4), \]
		which in turn leads to the Taylor expansion of $\beta_c(h)$ in the small field limit. 
		
		By inspection of \eqref{eq:bcrit2} as 	$h \to \infty $, the critical inverse temperature $\beta_c(h)$ tends to zero, but we still have that $h \beta_c(h) \to \infty$. 
		Moreover, we recall that $  \tanh(y) = 1- 2 e^{-2y} + \mathcal{O}(e^{-4y}) $ 
		 for large $y$ and $ r(1-x) = \frac12 x \ln(1/x)  + \mathcal{O}(x) $ for small $x$.
		After some algebra, we arrive at the asymptotic equation
		$ 2 \beta_c(h) h e^{2 \beta_c(h) h} = 8 h^2+ \mathcal{O}(h) $. 
		In particular,
		\[ \lim_{h \to \infty} \frac{2 \beta_c(h) h}{W(8h^2)} = \lim_{h \to \infty} \frac{ \beta_c(h) h}{\ln h} =1, \]
		where W denotes Lambert W-function.\\

		\noindent 
		2.~We first consider the high temperature limit. For small $\beta > 0$ a Taylor expansion yields
		\[ \arcosh\left(\frac12 \exp(\Phi^{\mathrm{REM}}(\beta,h))\right) =  \arcosh\left(1+ \frac{1}{2} \beta^2(1+h^2) + \mathcal{O}(\beta^4) \right) = \sqrt{1+h^2}\beta + \mathcal{O}(\beta^2),  \]
		from which we conclude $\Gamma_c(0,h) = 1$. As the term $\arcosh\left(\frac12 \exp(\Phi^{\mathrm{REM}}(\beta,h))\right)/\beta$ converges to the absolute value of the ground state energy as $\beta \to \infty$, we obtain the claim concerning the low temperature limit.\\
		
		\noindent 
		3.~Let us fix some $\beta >0$. We show that 
		\[ g(h) = \arcosh\left(\frac12 \exp(\Phi^{\mathrm{REM}}(\beta,h))\right)^2 - \beta^2 h^2   \]
		is strictly increasing which is equivalent to the monotonicity of $\Gamma_c(\beta,h)$. We compute the derivative for $h > 0$
		\[ \begin{split}
		g^\prime(h) &= 2 \arcosh\left(\frac12 \exp(\Phi^{\mathrm{REM}}(\beta,h))\right) \frac{\frac12 e^{\Phi^{\mathrm{REM}}(\beta,h)}}{\sqrt{\frac14 e^{2\Phi^{\mathrm{REM}}(\beta,h)}-1}} \frac{\partial \Phi^{\mathrm{REM}}(\beta,h)}{\partial h} - 2 \beta^2 h \\ & = 2 \beta  \left(\arcosh\left(\frac12 \exp(\Phi^{\mathrm{REM}}(\beta,h))\right) \frac{\tanh(\min\{\beta,\beta_c(h) \} h)}{\tanh(\arcosh(1/2\exp(\Phi^{\mathrm{REM}}(\beta,h))))} - \beta h \right)
		\end{split} \]
		We first note that $y/\tanh(y)$ is an increasing function. In the case $\beta \leq \beta_c(h)$ we further use that
		$1/2\exp(\Phi^{\mathrm{REM}}(\beta,h)) > \cosh(\beta h) $. 
		Hence $g'(h) > 0$ is an easy consequence of these observations for $\beta \leq \beta_c$.  On the other hand, if $\beta > \beta_c$ we use the  convexity of 
		\[ \eta(y) \coloneqq \frac{\arcosh(e^y)}{\tanh(\arcosh(e^y))},  \]
		from which we obtain 
		\[ \frac{\arcosh(\frac12 \exp(\Phi^{\mathrm{REM}}(\beta,h)) }{\tanh(\arcosh(1/2\exp(\Phi^{\mathrm{REM}}(\beta,h))))} > \frac{\beta h}{\tanh(\beta_c(h)h)}  \]
		as the left half side is a convex function of $\beta$ and the inequality holds true for $\beta = \beta_c(h)$.
		
		Finally, we want to show the asymptotic formula for $\Gamma_c(\beta,h)$. Since $\beta_c(h)$ tends to zero, we only need to consider the "frozen" expression for $\Phi^{\mathrm{REM}}(\beta,h)$. Neglecting terms of subleading order, we may write after some manipulations
		\[ \beta^{-2} \arcosh\left(\frac12 \exp(\Phi^{\mathrm{REM}}(\beta,h))\right)^2 - h^2 \simeq 2 h^2 (\tanh(\beta_c(h)h) - 1) + 2 \beta_c h . \]
		We recall that $ 1- \tanh(\beta_c(h)h) \simeq 2 e^{-2 \beta_c(h)h}  = \frac{4 \beta_c h}{2 \beta_c he^{2 \beta_c(h)h}   } 
		\simeq \frac{\beta_c}{2h} $, 
		where the last equality follows from the proof of part 1. Combining these asymptotic formulas, we arrive at 
		$ \lim_{h \to \infty} \frac{\Gamma(\beta,h)}{\sqrt{h \beta_c(h)}} = 1 $.
\end{proof}

\section{Proof of Proposition~\ref{prop:bsk} and Corollary~\ref{cor:hG}}\label{sec:pfsk}
In this section, we sketch the computations which lead to the results in  Proposition~\ref{prop:bsk} and Corollary~\ref{cor:hG}.

\begin{proof}[Proof of Proposition~\ref{prop:bsk} ]
1.~Let us first recall that $\bar{a}(x)$ is a continuous decreasing function from which it follows that 
	$ x(\beta) = \sup \{x|\,\bar{a}(x) > (2 \ln2)/\beta^2 \}  $
	is well defined for $\beta > \beta_c(0) = \sqrt{2 \ln2/ \bar{a}(0)}$ and increasing in $\beta$. Since $k$ is a decreasing function, we see that $\beta_c(h)$ defined in \eqref{eq:defcritbeta} is an increasing function.
	
	To discuss the limiting value $h \to 0$, we observe that
	$ \lim_{h \to 0} k(2 \ln2/(\beta_c(h)h)) = 0 $. 
	Since $ \bar{a} $ is continuous, $\lim_{\beta \to \beta_c} x(\beta) = 0 $ from which 
	we conclude $\lim_{h \to 0} \beta_c(h) = \beta_c(0)$. 
	Using Assumption~\ref{ass:a} we see that 
	\[ x(\beta) \propto (\beta-\beta_c)^{1/\alpha}.\]
	A direct calculation shows $ k(x)  \propto x^{-2} $ for large $x$. 
	We thus arrive at 
	$ \beta_c(h) - \beta_c(0) \propto h^{2 \alpha} $, 
	and $T_c - T_c(h) \propto h^{2 \alpha}$. \\
	For the limit $h \to \infty$, we first consider the case $\bar{a}(1) > 0$. Then, $x(\beta)$ approaches $1$ as $\beta \to \beta_c(\infty) \coloneqq \sqrt{2 \ln2/ \bar{a}(1)}$ and 
	\[ \lim_{h \to \infty} k(2 \ln2/(\beta_c(h)h)) = 0.  \]
	Consequently, $\lim_{h \to \infty} \beta_c(h) = \beta_c(\infty)$. Similarly, if $\bar{a}(1) = 0$, $x(\beta)$ approaches $1$ as $\beta \to \infty$ and we have  $\lim_{h \to \infty} \beta_c(h) = \infty$. \\
	
	\noindent
	2.a)~The continuity of $y(\beta,h)$ follows from the fact that it is a solution of a continuous implicit equation. Moreover, as $\phi(\beta,y)$ is decreasing in $y$ and k is a decreasing function,too, it follows from \eqref{eq:y} that $y(\beta,h)$ is increasing in $h$. As in part 1., one easily sees that 
	\[ \lim_{h \to 0} k(\varphi(\beta,y(\beta,h))/(\beta h)) = 0 \quad \text{and} \quad \lim_{h \to \infty}  k(\varphi(\beta,y(\beta,h))/(\beta h)) = 1, \]
	which in turn implies 
	$ \lim_{h \to 0} y(\beta,h) = 0 $ and $ \lim_{h \to \infty} y(\beta,h) =1 $.\\
	For the Taylor expansion, we use the fact that 
	\[ k(1/x) = \frac{\ln 2}{2} x^2 + \mathcal{O}(x^4). \]
	Consequently, we have 
	\[ y(\beta,h) = \frac{\ln 2}{2} \left(\frac{\beta h}{\varphi(\beta,y(\beta,h))}\right)^2 + \mathcal{O}(h^4) = \frac{\ln 2}{2} \left(\frac{\beta h}{\varphi(\beta,0)}\right)^2 + \mathcal{O}(h^4). \]
	Recalling that  
	\[ \varphi(\beta,0) = \begin{cases} \frac{\beta^2}{\ln2 \beta_c^2} + \ln 2 & \beta < \beta_c, \\ \frac{2 \ln 2 \beta}{\beta_c} & \beta \geq \beta_c,
	\end{cases}  \]
	we arrive at \eqref{asympy}. \\
	
	\noindent
	2.b)~Both assertions follow immediately from part 2a) and the fact that $\varphi(\beta,x)$ is continuous and decreasing in $x$.
\end{proof}

Finally, we present the proof of Corollary~\ref{cor:hG}:
\begin{proof}[Proof of Corollary~\ref{cor:hG} ]
The limit of the pressure is given by 
\begin{equation*}
		\Phi(\beta,\mathfrak{b},h) = \sup_{0 \leq y \leq z \leq 1} \left[ \beta h \gamma(y) + \int_{0}^{z-y} \varphi^{(y,1)}(\beta, x) \, dx + (1-z) p(\beta,\Gamma) \right] .
\end{equation*}
It follows that if $y(\beta,h) < z(\beta,\Gamma)$, then $y(\beta,h)$ and $z(\beta,\Gamma)$ remain the maximizer for this more general problem. We see that this holds true if and only if 
$ p(\beta\Gamma) < \varphi(\beta,y(\beta,h))$ and the pressure is then given by
\[ \Phi(\beta,\Gamma,h) = \beta h \gamma\left( y(\beta,h) \right) + \int_{ y(\beta,h)}^{z(\beta, \Gamma)} \varphi(\beta,x) dx + \left(1-z(\beta, \Gamma)\right) p(\beta\Gamma).   \]
Otherwise we have $y(\beta,h) \geq z(\beta,\Gamma)$ and, consequently, the corresponding maximizer satisfy $y^{\star}= z^{\star}$, i.e. 
\begin{equation*}
	\Phi(\beta,\Gamma,h) = \sup_{0 \leq y \leq 1} \left[ \beta h \gamma(y) + (1-y) p(\beta \Gamma) \right] .
\end{equation*}
Differentiating with respect to $y$ yields the maximizer 
\[ y^{\star} = \sigma(\beta,\Gamma,h) = k\left(\frac{p(\beta \Gamma)}{\beta h}\right) \]
since $k$ was defined to be the inverse of $\gamma^\prime$.
This completes the proof.
\end{proof}

 \minisec{Acknowledgments}
This work was supported by the DFG under EXC-2111 -- 390814868.

	\bigskip
	\bigskip
	\begin{minipage}{0.5\linewidth}
		\noindent Chokri Manai and Simone Warzel\\
		MCQST \& Zentrum Mathematik \\
		Technische Universit\"{a}t M\"{u}nchen
	\end{minipage}

\end{document}